\documentclass[12pt]{amsart}

\usepackage[a4paper,margin=23mm,top=30mm]{geometry}

\usepackage{amsfonts, amsmath, amssymb, amsgen, amsthm, amscd}
\usepackage{newtxtext,newtxmath}
\usepackage[utf8]{inputenc} 

\usepackage{color}

\usepackage[all]{xy}

\usepackage{hyperref}
\usepackage{mathtools,slashed}
\usepackage{setspace}
\usepackage{enumerate}

\def\a{\alpha}
\def\b{\beta}

\def\t{\theta}

\def\ve{\varepsilon}

\def\lra{\longrightarrow}

\def\ot{\otimes}

\def\lra{\longrightarrow}

\def\rt{\triangleright}
\def\lt{\triangleleft}

\def\Ad{\mathop{\rm Ad}\nolimits}
\def\ad{\mathop{\rm ad}\nolimits}

\def\dt{\left.\frac{d}{dt}\right|_{_{t=0}}}

\usepackage{enumerate}

\newcommand{\G}[1]{\mathfrak{#1}}
\newcommand{\C}[1]{\mathcal{#1}}
\newcommand{\B}[1]{\mathbb{#1}}

\newcommand{\ie}{{\it i.e.\/}\ }

\renewcommand{\leq}{\leqslant}
\renewcommand{\geq}{\geqslant}

\numberwithin{equation}{section}

\newtheorem{theorem}{Theorem}[section]
\newtheorem{proposition}[theorem]{Proposition}

\theoremstyle{definition}

\newtheorem{remark}[theorem]{Remark}
\newtheorem{example}[theorem]{Example}

\usepackage[all]{xy} 

\title{Matched pairs of discrete dynamical systems}

\author{O\u{g}ul Esen}
\address{Department of Mathematics, Gebze Technical University,  41400 Gebze-Kocaeli, Turkey}
\email{oesen@gtu.edu.tr}

\author{Serkan Sütlü}
\address{Department of Mathematics, I\c{s}ik University, 34980 \c{S}ile-\.{I}stanbul, Turkey}
\email{serkan.sutlu@isikun.edu.tr}

\begin{document}

\maketitle

\begin{abstract}
Matched pairs of Lie groupoids and Lie algebroids are studied. Discrete Euler-Lagrange equations are written for the matched pairs of Lie groupoids. As such, a geometric framework to analyse a discrete system by decomposing it into two mutually interacting subsystems is established. Two examples are provided to illustrate this strategy; the discrete dynamics on the trivial groupoid, and the discrete dynamics on the special linear group. 


\end{abstract}

\section{Introduction}

Consider two dynamical/mechanical systems, and their equations of motion presented either in Lagrangian or in Hamiltonian framework. Let also the systems be mutually interacting, so that they cannot preserve their individual motions. As a result, the equation of motion of the coupled system demands more effort than merely putting together the equations of motions of the individual systems. Such systems have been first studied within the semidirect product theory \cite{n2001lagrangian,MaMiOrPeRa07,MarsRatiWein84}, where only one of the systems is allowed to act on the other. Many physical systems fit into this geometry; such as the heavy top \cite{Rati80}, and the Maxwell-Vlasov equations \cite{holm1998euler}. 

The present paper is a part of the project which investigates a purely geometric framework (via a purely algebraic strategy) for decoupling the (discrete) dynamical equations of a system. More precisely, we strive to realize the dynamical equations of a system by means of the dynamical equations of two simpler systems, together with the additional terms that emerge from the mutual interactions of these, \cite{EsSu16,EsenSutl17}. This theory of ``matched pair dynamics'' may also be regarded as a generalization of the semi-direct product theory, which, in particular, corresponds to the case that one of the mutual representations being trivial.

The fundamental geometric object in the matched pair dynamics is that of a Lie group, whereas the algebraic strategy we follow is nothing but the matched pair theory of \cite{Ta81,Maji90-II}. A pair of Lie groups that act on each other, subject to compatibility conditions ensuring a group structure on their cartesian product, is called a ``matched pair of Lie groups''. We shall then call the total space (the cartesian product) ``the matched pair Lie group'' in order to emphasize the ``matching'', while it is also referred as a bicrossedproduct group in \cite{Majid-book,Maji90}, the twilled extension in \cite{KoMa88}, the double Lie group in \cite{LuWe90}, or the Zappa-Sz\'ep product in \cite{Br05}; see also \cite{Sz50,Sz58,szep1962sulle,Za42} . Conversely, from the decomposition point of view, if a Lie group is isomorphic (as topological sets) to the cartesian product of two of its subgroups (with trivial intersection), then it is a matched pair Lie group. In this case, the mutual actions of the subgroups are derived from the group multiplication.

Motivated by the fact that the tangent space of a Lie group has the structure of a Lie group (called the ``tangent group''), and that the tangent group of a matched pair Lie group is a matched pair tangent group, we successfully studied in \cite{EsenSutl17} the equations of motion of systems whose configuration spaces being matched pair tangent groups within the Lagrangian framework. We thus obtained the matched Euler-Lagrange equations, and the matched Euler-Poincar\'{e} equations.

It was expected then, to carry out the similar discussions along the cotangent bundles and the ``cotangent groups'' of Lie groups. This allowed us to apply in \cite{EsSu16} the algebraic machinery (of matched pairs) to the Hamiltonian formalism of the equations of motion. The matched Hamilton's equations, and the matched Lie-Poisson equations were obtained this way. The ``matched pair Hamiltonian dynamics'' found an application even in the field of fluid dynamics, \cite{esen2017hamiltonian}. 

On the other hand, the transition from the Kleinian conception of geometry to the Ehresmannian point of view led to the evolution of Lie groups into Lie groupoids, \cite{Para66,Prad07}, which may be considered as the central objects in ``discrete dynamics'', \cite{MaWe01,MarrDiegMart06}. In this geometry, a discrete system is generated by a Lagrangian on a Lie groupoid, and the dynamical equations are obtained by the directional derivatives of the Lagrangian with left and right invariant vector fields at a finite sequence of ``composable'' elements. 

The theory of discrete dynamics on the Lie groupoid framework has been studied extensively in the literature. We refer the reader to \cite{IgMADiMa08,MaMaSt15} for the discrete dynamics involving constraints, and to \cite{VaCa07} for the field theoretic approach. In particular, there is an even richer theory of discrete dynamics on Lie groups. We cite \cite{BoSu99} for the discrete time Lagrangian mechanics on Lie groups, \cite{MaPeSh99} for the discrete Lie-Poisson and the discrete Euler-Poincar\'e equations, and \cite{MaPeSh2000,JaSaLeMa06} for the reduction of discrete systems under symmetry. For the discrete Hamiltonian dynamics in the realm of variational integrators we refer to \cite{leok2010discrete}, and \cite{colombo2012discrete,colombo2013higher}  for the higher order discrete dynamics. The local description of discrete mechanics can be found in \cite{Marrero2015local}, whereas the inverse problem of the calculus of variations in \cite{Barbero2018inverse}. An implicit formulation of the discrete dynamics has been introduced in \cite{Iglesias2010discrete}.

In much the same way the Lie groups are matched, the Lie groupoids (as well as their Lie algebroids) can be matched. As a result, the very same algebraic strategy of ``matched pairs'' may be used once more; this time to study the discrete dynamical systems. What we achieve then, in the present paper, is to represent the equations of motion of such a system in terms of the equations of motion of two simpler systems, decorated by the additional terms reflecting the mutual interactions of the subsystems. 

The paper is organized in four sections. In the following two sections, Section \ref{sect-Lie-groid-Lie-algoid} and Section \ref{matched-Lie-groids-Lie-algoids}, we recall what we shall use on Lie groupoids, Lie algebroids, and their matched pair theory. It is Section \ref{sect-disc-dynm-matched-pair} where the novelty of the paper lies. The discerete Euler-Lagrange equations are recalled in Subsection \ref{subsect-disc-Euler-Lagrange-eqns}, and then revisited both in Subsection \ref{subsect-disc-dynm-on-matched-pair-grpoid} for matched pairs of Lie groupoids, and in Subsection \ref{subsect-disc-dynm-matched-Lie-grps} for the Lie groups; along the way towards the discrete dynamics on the matched pairs of Lie groups. Section \ref{ex} is reserved for concrete examples. More precisely, in Subsection \ref{ex1} we present the (discrete) Euler-Lagrange equations explicitly for the trivial groupoid, decomposing it into the action groupoid and the coarse groupoid. As such, we realise the (discrete) Euler-Lagrange equations of the trivial groupoid in terms of the (discrete) Euler-Lagrange equations of the action groupoid and the coarse groupoid, glued together via the terms that emerge from the mutual actions of those. Finally, in Subsection \ref{ex2} we illustrate the (discrete) Euler-Lagrange equations on the Lie group $SL(2,\mathbb{C})$, the matched pair decomposition of which incarnates as the Iwasawa decomposition.

\section{Lie groupoids and Lie algebroids}\label{sect-Lie-groid-Lie-algoid}

In order to fix the notation, as well as the convenience of the reader, we devote the present section to a brief summary of the basics of Lie groupoids and Lie algebroids. The reader may consult to \cite{Mackenzie-book,Mackenzie-book-II,MoerMrcu10} for further details.

\subsection{Lie groupoids and their actions}

\subsubsection{Definition and basic examples}~

\noindent
Let $\C{G}$ and $B$ be two manifolds, and let there be two surjective submersions
\[
\xymatrix{
\C{G} \ar@<1ex>[r]^{\a} \ar@<-1ex>[r]_{\b} & B,
}
\]
called the ``source map'' and the ``target map'', respectively. We assume also that, there exists a smooth map 
\[
\ve:B \lra \C{G}, \qquad b \mapsto \widetilde{b},
\]
called the ``object inclusion''. The product space
\[
\C{G} \ast \C{G} :=\big\{ (g,g') \in \C{G} \times \C{G} \mid \b(g) = \a(g')\big\}
\]
is called the ``space of composable elements'', and is equipped with the partial multiplication 
\[
\C{G} \ast \C{G} \lra \C{G}, \qquad (g,g') \mapsto gg'.
\]
The five-tuple $(\C{G},B,\a,\b,\epsilon)$ with a partial multiplication is called a ``Lie groupoid'' if    
\begin{itemize}
\item[(i)] $\a(gg') = \a(g)$, and $\b(gg') = \b(g')$, 
\item[(ii)] $g(g' g'') = (g g') g''$, 
\item[(iii)] $\a(\widetilde{b}) = \b(\widetilde{b}) = b$,
\item[(iv)] $g \widetilde{\b(g)} = g = \widetilde{\a(g)} g$ ,
\item[(v)] there is $g^{-1} \in \C{G}$ such that $\a(g^{-1}) = \b(g)$ and $\b(g^{-1}) = \a(g)$, and that
\[g^{-1} g = \widetilde{\b(g)}, \qquad gg^{-1} = \widetilde{\a(g)},\]
\end{itemize}
for any $(g,g'),(g',g'') \in \C{G} \ast \C{G}$, any $b \in B$, and any $g \in \C{G}$.

The elements of $B$ are called  ``objects'', whereas the elements of $\C{G}$ are referred as ``arrows'', or ``morphisms'' . A groupoid may also be considered as a category such that all arrows are invertible. We shall denote a Lie groupoid by $\C{G}\rightrightarrows B$, or simply by $\C{G}$ when there is no confusion on the base. Let us, now, recall the examples of Lie groupoids that we shall need in the sequel.

\begin{example}\label{ex-Lie-gr-as-grpd}
Any Lie group $G$ gives rise to a Lie groupoid over the identity element $\{e\}$; the source map and the target map being the constant maps $\a = \b: G \to \{e\}$, and the object inclusion map being the obvious inclusion $\widetilde{e} = e$. Moreover, the partial multiplication of this groupoid is the group multiplication. We denote this Lie groupoid by $G\rightrightarrows \{e\} $, or simply by $G$.
\end{example}

\begin{example}\label{ex-action-groupoid}
Let $G$ be a Lie group, and $M$ a manifold with a smooth $G$-action $M\times G \to M$ from the right. Then $M\times G\rightrightarrows M$ has the structure of a Lie groupoid over $M$ equipped with the source map, the target map, and the object inclusion given by
\begin{eqnarray*}
\a &:& M\times G \lra M, \qquad \a(m,g) := m,
\\
\b &:& M\times G \lra M, \qquad \b(m,g) := m g,
\\
\ve&:& M\lra M\times G, \qquad \ve(m) := (m,e).
\end{eqnarray*}
The partial multiplication is given by
\begin{equation}
(m,g) \cdot (m',g') := (m, gg'), \quad {\rm if} \quad m g = m'.
\end{equation}
The groupoid $M\times G$ is called as the ``action groupoid''. 
\end{example}

\begin{example}\label{ex-coarse-groupoid}
Let $M$ be a manifold. Then the cartesian product $M \times M$  of $M$ with itself is a groupoid over $M$ via the source, target, and the object inclusion  maps given by
\begin{eqnarray*}
\a &:& M\times M \lra M, \qquad \a(m,m') := m, \\
\b &:& M\times M \lra M, \qquad \b(m,m') := m',\\
\ve &:& M\lra M \times M, \qquad \ve(m) := (m,m),
\end{eqnarray*}
whereas the partial multiplication is
\[
(m,m') \cdot (n,n') := (m,n'), \quad {\rm if} \quad m' = n.
\]
This groupoid $M \times M \rightrightarrows M$ is called the ``coarse groupoid'', the ``pair groupoid'', or the ``banal groupoid''.
\end{example}

\begin{example}\label{ex-trivial-groupoid}
Given a manifold $M$, and a Lie group $G$, the triple product $M \times G \times M$ is a Lie groupoid over $M$ by the source, target, and the object inclusion  maps
\begin{eqnarray*}
\a &:& M \times G\times M \lra M, \qquad \a(m,g,m') := m,
\\
\b &:& M\times G\times M \lra M, \qquad \b(m,g,m') := m',
\\
\ve &:& M\lra M \times G\times M, \qquad \ve(m) = \widetilde{m} := (m,e,m),
\end{eqnarray*}
and the partial multiplication
\[
(m,g,m') \cdot (n,g',n') := (m,gg', n'), \quad {\rm if} \quad m' = n. 
\]
The groupoid $M \times G \times M \rightrightarrows M$ is called the ``trivial groupoid''.
\end{example}

\subsubsection{Left and right invariant vector fields}~

\noindent
Given a Lie groupoid $\C{G}\rightrightarrows B$, a vector field $Z \in \Gamma(T\C{G})$ is called ``left invariant'' if 
\[
Z(gg') = T_{g'}\ell_g Z(g'), 
\]
for any $(g,g') \in \C{G}\ast \C{G}$, where $\ell_g:\C{G}\to \C{G}$ is the left translation induced from the partial multiplication, and   $T_{g'}\ell_g:T_{g'}\C{G}\to T_{gg'}\C{G}$ is the tangent lift of this mapping at $g' \in \C{G}$. Similarly, a ``right invariant'' vector field $Z\in \Gamma(T\C{G})$ is one that satisfies 
\[
Z(gg') = T_gr_{g'}Z(g)
\]
for any $(g,g') \in \C{G}\ast \C{G}$, where $r_{g'}:\C{G}\to \C{G}$ is the right translation, and $T_gr_{g'}:T_g\C{G}\to T_{gg'}\C{G}$ is its the tangent lift at $g\in \C{G}$.

\subsubsection{Morphisms of Lie groupoids}~

\noindent
Let $\C{G}$ an $\C{H}$ be two Lie groupoids over the bases $B$ and $C$, respectively. A morphism of Lie groupoids is a pair of smooth maps $\Phi: \C{G} \lra \C{H}$ and $\Phi_0:B\lra C$ compatible with the groupoid multiplications, source, target, and the object inclusion maps. More precisely, a Lie groupoid morphism is a pair $(\Phi,\Phi_0)$ that satisfies
\begin{itemize}
\item[(i)] $(\Phi(g), \Phi(g')) \in \C{H}\ast \C{H}$, 
\item[(ii)] $\Phi(g g') = \Phi(g) \Phi(g')$,
\item[(iii)] $ \a (\Phi(g)) = \Phi_0(\a (g))$,
\item[(iv)] $\b (\Phi(g)) = \Phi_0(\b(g))$,
\item[(v)] $\Phi(\widetilde{b}) = \widetilde{\Phi_0(b)}$,
\end{itemize}
for any $(g,g') \in \C{G}\ast \C{G}$, and any $b \in B$. The requirements may be summarized by the commutativity of the diagram  
\begin{equation} \label{Lgmorp}
\xymatrix{ \C{G} \ar@<1ex>[d]^{\a} \ar@<-1ex>[d]_{\b}
\ar [rr]^{\Phi}&& \C{H}  \ar@<1ex>[d]^{\a} \ar@<-1ex>[d]_{\b}\\ B \ar@/^2pc/[u]^{\ve}
\ar [rr]_{\Phi_0}&& C \ar@/_2pc/[u]_{\ve}
}
\end{equation} 
for each source, target and inclusion map. In order to avoid the notation inflation, we shall not distinguish the source maps of the Lie groupoids  $\C{G}$ and $\C{H}$ with different notations. Instead, we shall make it clear from the context.

\subsubsection{Lie groupoid actions}\label{subsubsect-Lie-grpd-action}~

\noindent
Let $\C{G}$ be a Lie groupoid over the base $B$, and let $f:P\to B$ be a smooth map from a manifold $P$ to the base manifold $B$. Given the product space
\[
P \ast \C{G} := \big\{(p, g) \in P \times \C{G} \mid f(p) = \a(g)\big\},
\]
a smooth map 
\[
\lt: P \ast \C{G} \lra P, \qquad (p,g) \mapsto p \lt g
\]
is called the (right) action of $\C{G}$ on $f$ if 
\begin{itemize}
\item[(i)] $f(p \lt g) = \b(g)$,
\item[(ii)] $(p \lt g) \lt g' = p \lt (gg')$, 
\item[(iii)] $p \lt \widetilde{f(p)} = p$,
\end{itemize}
for any $(p,g) \in P \ast \C{G}$, any $(g,g') \in \C{G} \ast \C{G}$, and any $p \in P$. The definition may be summarized by the commutativity of the diagram
\begin{equation}
\xymatrix{ P: & p \ar [d]_{f} \ar [rr]_{ \lt g} \ar@/^2pc/[rrrr]^{\lt gg'}
&& p \lt g \ar [d]_{f}    \ar [rr]_{ \lt g'} && (p \lt g) \lt g' \ar [d]_{f}  \\ B: &
\a(g)\ar [rr]_{g}\ar@/_2pc/[rrrr]_{gg'}&& \b(g)=\a(g') \ar [rr]_{g'}&&\b(g')
} 
\end{equation} 
The letters $P$ and $B$ refers to the manifolds in which the objects in the corresponding rows belong. 

The left action of a Lie groupoid on a smooth map is defined similarly, \cite{Mackenzie-book,Mack92}. Let $\C{H}$ be a Lie groupoid over the base $B$, and let  $f:P\mapsto B$ be a smooth function from a manifold $P$ to $B$. Then, given the product space
\[
\C{H}\ast P := \big\{(h, p) \in \C{H} \times P \mid \b(h) = f(p)\big\},
\]
the smooth mapping
\[
\rt: \C{H}\ast P \lra P, \qquad (h,p) \mapsto h \rt p
\]
is called a (left) action of $\C{H}$ on $f$ if 
\begin{itemize}
\item[(i)] $f(h \rt p) = \a(h)$,
\item[(ii)] $h' \rt (h \rt p) = (h'h) \rt p$, 
\item[(iii)] $\widetilde{f(p)}\rt p  = p$,
\end{itemize}
 for any $(h,p) \in \C{H} \ast P$, any $(h',h) \in \C{H} \ast \C{H}$, and any $p \in P$. In other words, the following diagram is commutative:
\begin{equation}
\xymatrix{ P: &h'\rt (h \rt p) \ar [d]_{f} 
&& \ar [d]_{f}  \ar [ll]_{ h'\rt} h \rt   p && \ar [d]_{f} \ar [ll]_{ h\rt}p \ar@/_2pc/[llll]_{h'h\rt}\\ B: & 
\a(h')\ar [rr]_{h'}\ar@/_2pc/[rrrr]_{h'h}&& \b(h')=\a(h) \ar [rr]_{h}&&\b(h)
} 
\end{equation} 

We are now ready to conclude with the left and the right actions of a Lie groupoid on another Lie groupoid. Let $\C{G}$ be a Lie groupoid over the base $B$, and let $\C{H}$ be another Lie groupoid over $C$. The left action of $\C{H}$ on $\C{G}$ is defined to be the (left) action of  $\C{H}$ on the source map $\a:\C{G}\to B$, while the right action of $\C{G}$ on $\C{H}$ is similarly defined to be the (right) action of $\C{G}$ on the target map $\b:\C{H}\to C$. For further details on the representations of Lie groupoids we refer the reader to \cite{Brow72,Mack92}. 

\subsection{Lie algebroids and their actions}

\subsubsection{Definition of a Lie algebroid}~

\noindent
A Lie algebroid over a manifold $M$ may be thought of a generalization of the tangent bundle $TM$ of $M$, \cite{Mackenzie-book,Mokr97,Para67}. More technically, given a manifold $M$, a ``Lie algebroid'' $\C{A}$ over the base $M$ is a (real) vector bundle $\tau:\C{A} \to M$, together with a map $a:\C{A} \to TM$ of vector bundles, called the ``anchor map'', and a Lie bracket $[\bullet,\bullet]$ (bilinear, anti-symmetric, satisfying the Jacobi identity) on the space $\Gamma(\C{A})$ of sections, so that the induced $C^\infty(M)$-module homomorphism $a:\Gamma(\C{A})\to \Gamma(TM)$ satisfies
\[
[X,fY] = f[X,Y]+\C{L}_{a(X)}(f)Y
\]
for any $X,Y \in \Gamma(\C{A})$, and any $f\in C^\infty(M)$, where $\C{L}_{a(X)}(f)$ stands for the directional derivative of $f\in C^\infty(M)$ in the direction of $a(X)\in TM$. It, then, follows that 
\[
a([X,Y]) = [a(X),a(Y)]
\]
for any $X,Y \in \Gamma(\C{A})$. Accordingly, a Lie algebroid is denoted by a quintuple $(\C{A},\tau,M,a,[\bullet,\bullet])$, or occasionally by a triple $(\C{A},\tau,M)$ when there is no confusion on the bracket, and the anchor map. Let us now take a quick tour on a bunch of critical examples.

\begin{example} \label{LieAlgebra}
Any Lie algebra $\G{g}$ is a Lie algebroid $\tau:\G{g}\to \{\ast\}$; taking the base manifold $B=\{\ast\}$ to be a one-point set, and the anchor map $a:\G{g}\to TB$ to be the zero map. 
\end{example}

\begin{example}\label{ex-TM-Lie-algoid}

Given any manifold $M$, the tangent bundle $TM$ is a Lie algebroid over the base $M$, where the anchor $a:TM\to TM$ is the identity map.
\end{example}

\begin{example}\label{ex-action-Lie-algoid}
Let $M$ be manifold admitting an infinitesimal left action of a Lie algebra $\G{g}$. As such, there exists a linear map $\G{g}\to \Gamma(TM)$, $\xi\mapsto X_\xi$, preserving the Lie brackets. Then the trivial bundle $M\times \G{g} \to M$, via the projection onto the first component, can be made into a Lie algebroid via the (fiber preserving) anchor map 
\[
a: M\times \G{g} \to TM, \qquad a(m, \xi):=X_\xi(m),
\]
and the bracket given by 
\[
[u,v](m) = [u(m),v(m)]_\G{g} + (\C{L}_{X_{u(m)}}v)(m) - (\C{L}_{X_{v(m)}}u)(m)
\]
regarding the sections of the trivial bundle $M\times \G{g}$ as (smooth) maps $u,v:M\to \G{g}$, referring the Lie derivative of a Lie algebra valued function $w:M\to \G{g}$ along the vector field $X \in \Gamma(TM)$ by $\C{L}_X(w)$, and denoting the Lie bracket on $\G{g}$ as $[\bullet,\bullet]_\G{g}$. This Lie algebroid is called the ``action Lie algebroid'' (or the ``transformation Lie algebroid''), \cite{HiggMack90,KosmSchw-Mack02,Lu97,Mokr97}.
\end{example}

\subsubsection{Lie algebroid of a Lie groupoid}\label{susubsect-Lie-algbrd-of-a-lie-grpd}~

\noindent
We now recall from \cite{CrampinSaunders-book,Mackenzie-book-II,Para67} that how a Lie algebroid associates to a given Lie groupoid in a canonical way. Let $\C{G}$ be a Lie groupoid  over $B$, $\a:\C{G} \to B$ be the source map, and $T_g\a:T_g \C{G} \to T_{\a(g)}B$ be its tangent lift at $g\in \C{G}$. Along the lines of \cite{MarrDiegMart06}, the ``Lie algebroid associated to the Lie groupoid'' $\C{G}$ is defined to be the vector bundle $(\C{A}\C{G},\tau,B)$ whose fibers are given by  
\[
\C{A}_b\C{G} := \ker T_{\ve(b)}\a,
\]
In other words, $\C{A}\C{G}$ corresponds to the vertical bundle on $\C{G}$ with respect to the fibration $\a:\C{G}\rightarrow B$. We shall denote a typical section of the fibration $\tau:\C{A}\C{G}\to B$  by $X\in \Gamma(\C{A}\C{G})$. Moreover, the anchor map $a:\C{AG} \to TB$ is given by 
\begin{equation}
a(X(b))=T_{\widetilde{b}}\b \circ X (b)
\end{equation}
where $T_{\widetilde{b}}\b:T_{\widetilde{b}}\C{G}\to T_bB$ is the tangent lift of the target map $\b:\C{G}\to B$ at $\widetilde{b}=\ve(b) \in \C{G}$. Finally, the definition of  the Lie bracket 
$[\bullet,\bullet]_{\C{AG}}$ on the sections of $\C{AG}\rightarrow B$ follows from the direct analogy with that of the Lie groups. More precisely, the bracket on the (associated) Lie algebroid is defined by means of the Jacobi-Lie bracket of the left (or the right) invariant vector fields on the groupoid; \cite{CostDazoWein87,Mackenzie-book,MarrDiegMart06}. 

For a Lie groupoid $\C{G}\rightrightarrows B$, there exists an isomorphism of $C^\infty(\C{G})$-modules  between the sections $\Gamma(\C{AG})$ of the Lie algebroid of the Lie groupoid and the left-invariant (resp. the right invariant) vector fields on $\C{G}$, \cite{Mackenzie-book-II}. To any $X\in \Gamma(\C{AG})$, there corresponds a left invariant vector field $\overleftarrow{X} \in \Gamma(T\C{G})$ via 
\begin{equation}\label{left-inv-vect-fields}
\overleftarrow{X}(g): = T_{\widetilde{\b(g)}}\ell_g X(\b(g)).
\end{equation}
As such, a left invariant vector field satisfies
\[
\overleftarrow{X}(g_1g_2) = T_{g_2}\ell_{g_1}\overleftarrow{X}(g_2),
\]
for any (composable) $g_1,g_2 \in \C{G}$. Conversely, if $\overleftarrow{X} \in \Gamma(T\C{G})$ is a left invariant vector field, then 
\[
X(b):=\overleftarrow{X}(\widetilde{b})
\]
defines a section $X\in \Gamma(\C{AG})$. This identification of the sections of the Lie algebroid $\C{AG}$ with the left invariant vector fields of $\C{G}$ allows us to define (induce) the Lie bracket on the sections. Given $X,Y \in \Gamma(\C{AG})$, we define their Lie bracket as
\begin{equation} \label{braAG}
[X,Y]_{\C{AG}} (b):= [\overleftarrow{X},\overleftarrow{Y}](\widetilde{b})
\end{equation}
where the bracket on the right hand side is the Jacobi-Lie bracket of left invariant vector fields on the Lie algebroid $\C{G}$. 

As for the right invariant vector fields, given any $X \in \Gamma(\C{AG})$, a right invariant vector field $\overrightarrow{X} \in \Gamma(T\C{G})$ is defined through 
\begin{equation}\label{right-inv-vect-fields}
\overrightarrow{X}(g) := - T_{\widetilde{\a(g)}}r_g \circ T_{\widetilde{\a(g)}}{\rm inv}(X(\a(g)),
\end{equation}
where ${\rm inv}:\C{G}\to \C{G}$ stands for the inversion, and $T_{\widetilde{\a(g)}}r_g:T_{\widetilde{\a(g)}}\C{G}\to T_g\C{G}$ is the tangent lift of the right translation by $g\in \C{G}$ at $\ve(\a(g)) = \widetilde{\a(g)} \in \C{G}$. 

The right invariant vector field $\overrightarrow{X} \in \Gamma(T\C{G})$ corresponding to $X \in \Gamma(\C{AG})$ is equivalently given by
\begin{equation} \label{right-inv-vect-fields-curve}
\overrightarrow{X}(g)=-\dt\,x^{-1}(t)g,
\end{equation}
where $x(t)$ is any curve through $\C{G}$ with constant source, that is $\a(x(t)) = \a(g) \in B$, and is tangent to $X(\a(g)) \in \C{A}_{\a(g)}\C{G} \subseteq T_{\widetilde{\a(g)}} \C{G}$, that is $\dot{x}(0)=X(\a(g))$.

It follows at once from the definition that 
\[
\overrightarrow{X}(g_1g_2) = T_{g_1}r_{g_2}\overrightarrow{X}(g_1)
\]
for any (composable) $g_1,g_2 \in \C{G}$. In this case, the relation between the Lie brackets is given by 
\[
 \overrightarrow{[X,Y]}_{\C{AG}} = - [\overrightarrow{X},\overrightarrow{Y}],
\]
where the bracket on the left side is the one on the Lie algebroid level, whereas the bracket on the right hand side is the Jacobi-Lie bracket of vector fields on the manifold $\C{G}$. 

Let us next present examples on the construction of the left and the right invariant vector fields on the Lie groupoids in Example \ref{ex-Lie-gr-as-grpd} - \ref{ex-trivial-groupoid} which will be needed in the sequel.

\begin{example} 
Let $G$ be a Lie group, and $G \rightrightarrows \{e\}$ be the Lie groupoid in Example \ref{ex-Lie-gr-as-grpd}. Since the inclusion map is given by $\ve(e)=\widetilde{e}=e$, the tangent space at $\widetilde{e}\in G$ is $T_e G= \mathfrak{g}$, and the source map $\a:G\to \{e\}$ is constant, the kernel of the tangent lift of the source map (the total space of the associated Lie algebroid $\C{A}G$) may be identified with the Lie algebra $\G{g}$ of the group $G$. The right and the left invariant vector fields associated to $\xi \in \G{g}$ are then given by
\begin{eqnarray} \label{R-L-inv-VF-Lie}
\overrightarrow{\xi}: G \lra TG, \qquad g\mapsto T_{e}r_{g}(\xi), \\ 
\overleftarrow{\xi}:G \lra TG, \qquad g\mapsto T_{e}\ell_{g}(\xi).
\end{eqnarray}
\end{example}

\begin{example}\label{ex-pair-grpd-left-right-inv-v-f}
The Lie algebroid of the coarse groupoid $M\times M \rightrightarrows M$ of Example \ref{ex-coarse-groupoid} is the Lie algebroid $(TM,\tau_M,M)$ of Example \ref{ex-TM-Lie-algoid}. Indeed, consider a curve $(m,n_t) \in  M\times M$ with constant source, so that $n_0 = n \in M$ and $\dot{n}_0=X \in T_nM$. Then its derivative (at $t=0$) yields a vector $(\t_m,X) \in T_mM \times T_nM$. Given any section $(\t_m,X) \in \C{A}_m(M\times M) $, the corresponding right and left invariant vector fields on $M\times M$ are computed to be
\begin{eqnarray} \label{R-L-inv-VF-Banal}
\overrightarrow{(\t_m,X) }&:&M\times M \lra T(M\times M) = TM\times TM, \qquad (m,n)\mapsto(-X,\t_m), \\ 
\overleftarrow{(\t_m,X) }&:&M\times M \lra T(M\times M) = TM\times TM, \qquad (m,n)\mapsto(\t_m,X).
\end{eqnarray}
\end{example}

\begin{example}\label{ex-action-grpd-left-right-inv-v-f}
The Lie algebroid of the action groupoid $M\times G \rightrightarrows M$ of Example \ref{ex-action-groupoid} is the transformation Lie algebroid $(M\times \G{g},pr_1,M)$ of Example \ref{ex-action-Lie-algoid}. Indeed, the derivative of a curve $(m,g_t) \in M\times G$ of constant source, such that $g_0=e$ and that $\dot{g}_0 = \xi \in \G{g}$, yields a vector $(\t_m,\xi) \in T_mM\times T_eG$. As such, the total space  $\C{A}(M\times G)$ may be identified with the cartesian product $M\times \G{g}$. For any  $(\t_m,\xi) \in \C{A}_m(M\times G) $, the left invariant and the right invariant vector fields are given by
\[
\overrightarrow{(\t_m,\xi)}\left( m,g\right) =\left( -\xi^\dagger(
m) ,\overrightarrow{\xi }\left(g\right) \right) ,\text{ \ \ }%
\overleftarrow{(\t_m,\xi)}\left(m,g\right) =\left(\t_m ,\overleftarrow{\xi 
}\left(g\right) \right) ,
\]
where the infinitesimal generator of the right action is computed to be
\begin{equation}\label{dagger}
\xi^\dagger(m) := \dt\,me_t
\end{equation}
for any curve $e_t\in G$ with $e_0=e\in G$, and $\dot{e}_0 = \xi \in \G{g}$.
\end{example}

\begin{example}\label{ex-triv-grpd-left-right-inv-v-f}
The Lie algebroid of the trivial groupoid of Example \ref{ex-trivial-groupoid}  is the Lie algebroid $(M\times \G{g}) \oplus TM$. Indeed, given a curve $(m_t,e_t,n_t) \in M\times G\times M$ with $e_0 = e \in G$, $m_0 = n_0 = m \in M$, $\dot{e}_0=\xi \in \G{g}$, $\dot{m}_0=X \in T_mM$, and $\dot{n}_0=Y \in T_mM$, 
\[
T_{(m,e,m)}\a (X,\xi,Y) = X.
\]
As such, 
\[
\C{A}_m(M\times G\times M) = \big\{(\t_m,\xi,Y)\mid \xi \in \G{g},\,Y\in T_mM\big\}
\]
which yields $\C{A}(M\times G \times M) \cong (M\times \G{g}) \oplus TM$.
In order to compute the left (resp. right) invariant vector fields, let $(m,g,n) \in M\times G \times M$, and let $(\t_n,\xi,X) \in \C{A}_n(M \times G \times M)$. Accordingly, let $(n,e_t,n_t) \in M\times G \times M$ be a curve such that $e_0= n$, $n_0=n$, $\dot{e}_0=\xi$, and $\dot{n}_0=Y\in T_nM$. We then see that
\begin{equation}\label{left-inv-(X,u)}
\overleftarrow{(\t_n,\xi,Y) }(m,g,n) = \dt\,(m,g,n)(n,e_t,n_t) = \dt\,(m,ge_t,n_t) =(\t_m, \overleftarrow{\xi}(g), Y).
\end{equation}
Similarly, given $(\t_m,\xi,X)$, if $(m,e_t,m_t) \in M\times G \times M$ be a curve such that $e_0= e$, $m_0=m$, $\dot{e}_0=\xi$, and $\dot{m}_0=X\in T_mM$, then
\begin{align*}
& \overrightarrow{(\t_m,\xi,X)} (m,g,n)= -\dt\, (m,e_t,m_t)^{-1}(m,g,n) = -\dt\, (m_t,e_t^{-1},m)(m,g,n) =\\
&  -\dt\, (m_t,e_t^{-1}g,n) = (-X,\overrightarrow{\xi}(g),\t_n).
\end{align*}
\end{example}


\subsubsection{Morphisms of Lie algebroids}~

\noindent
Given two Lie algebroids $(\C{A},\tau,M)$ and $(\C{A}',\tau',M')$, a morphism from $(\C{A},\tau,M)$ to $(\C{A}',\tau',M)$ is a vector bundle morphism preserving the achors as well as the brackets. That is, a pair $(\phi:\C{A} \to \C{A}'; \phi_0:M \to M')$ such that
\[
\xymatrix{
\C{A} \ar[r]^\phi\ar[d]_\tau & \C{A}' \ar[d]^{\tau'} \\
M\ar[r]_{\phi_0} & M',
}
\]
that
\[
a'\circ\phi = T\phi_0 \circ a,
\]
where $a:\C{A} \to TM$ and $a':\C{A}'\to TM'$ are the respective anchor maps, and finally that
\[
\phi([X,Y]) = [\phi(X),\phi(Y)].
\]

It is possible to derive a Lie algebroid morphism starting from a Lie groupoid morphism as follows. Given two Lie groupoids $\C{G}\rightrightarrows B$  and $\C{H}\rightrightarrows C$. Let $$\Phi:\C{G}\to \C{H}; \qquad \Phi_0:M \to N$$ be a morphism of Lie groupoids. Then, for $\C{A}_m\C{G} \ni X = \dt\, x_t$, where $\a(x_t) = m$,
\[
\C{A}_m\Phi:\C{A}_m\C{G} \to \C{A}_{\Phi_0(m)}\C{H};\qquad  X\mapsto \dt\,\Phi(x_t)
\]
defines a morphism $\C{AG}\to \C{AH}$ of Lie algebroids associated with the Lie groupoids $\C{G}$ and $\C{H}$, respectively. We refer the reader to \cite[Sect. 3.5]{Mackenzie-book} for further details.

\section{Matched pairs of Lie groupoids and matched pairs of Lie algebroids }\label{matched-Lie-groids-Lie-algoids}

\subsection{Matched pairs of Lie groupoids} \label{mplg}

\subsubsection{Definition of a matched Lie groupoid}~

\noindent
In this subsection we recall, mainly from \cite{Mokr97,Mack92}, the groupoid level of the matched pair theory of \cite{Maji90,Majid-book}. Let $\C{G}\rightrightarrows B$ and $\C{H}\rightrightarrows B$ be two Lie groupoids over the same base $B$, and let $\C{H}$ act on $\C{G}$ from the left by
\begin{equation} \label{rt}
\rt: \C{H} \ast \C{G} \lra \C{G}, \qquad (h, g') \mapsto h \rt g',
\end{equation}
where we recall that the set $\C{H} \ast \C{G}$ of composable elements consists of the pairs $(h,g)\in \C{H} \times \C{G}$ such that $\b(h)=\a(g)$. Being a left action, \eqref{rt} satisfies
\begin{itemize}
\item[(i)] $\a(h)=\a(h\rt g')$ ,
\item[(ii)] $(h'h) \rt g' =h' \rt (h \rt g')$
\item[(iii)] $\widetilde{\a(h)}\rt g' = g'$ for any  $h \in \C{H}$
\end{itemize}
for any  $(h, g')\in \C{H}\ast \C{G}$, any $(h',h) \in \C{H}\ast \C{H}$, and any  $h \in \C{H}$. Let also $\C{G}$ act on $\C{H}$ from the right by 
\begin{equation}\label{lt}
\lt: \C{H} \ast \C{G} \lra \C{H}, \qquad (h, g') \mapsto h \lt g'.
\end{equation}
Then, being a right action, \eqref{lt} satisfies
\begin{itemize}
\item[(iv)] $\b(g')=\b(h\lt g')$,
\item[(v)] $ h \lt g'g = (h \lt g') \lt g$, 
\item[(vi)] $ h\lt \widetilde{\b(g')} = h$,
\end{itemize}
for any  $(h, g')\in \C{H}\ast \C{G}$, any $(g',g) \in \C{G}\ast \C{G}$, and any  $g' \in \C{G}$. Now, the pair $(\C{G},\C{H})$ is called a ``matched pair of Lie groupoids'' if, in addition, the compatibilities 
\begin{itemize}
\item[(vii)] $\b(h \rt g') = \a(h \lt g')$,
\item[(viii)] $h \rt (g' g) = (h \rt g')  ((h \lt g') \rt g)$, 
\item[(ix)] $(h'  h) \lt g' = (h' \lt (h \rt g'))  (h \lt g')$, 
\end{itemize}
are also satisfied for any $(h,g') \in \C{H} \ast \C{G}$, any $(g', g) \in \C{G} \ast \C{G}$, and any $(h', h) \in \C{H} \ast \C{H}$. 



Then, the product space 
\[
\C{G} \bowtie \C{H}:=\C{G} \ast \C{H}  = \big\{(g,h) \in \C{G} \times \C{H} \mid \b(g)=\a(h)\big\}
\]
becomes a Lie groupoid by the partial multiplication, on
\[
(\C{G} \bowtie \C{H}) \ast (\C{G} \bowtie \C{H}):=\big\{((g,h),(g',h')) \in (\C{G} \bowtie \C{H})\times (\C{G} \bowtie \C{H})\mid \b(h) = \a(g')\big\},
\]
given by 
\begin{equation} \label{mpmul}
(\C{G} \bowtie \C{H}) \ast (\C{G} \bowtie \C{H}) \longrightarrow (\C{G} \bowtie \C{H}), \qquad
\left(  \left(  g,h \right) , \left(  g',h' \right)  \right) \mapsto \left(  g\left(  h\rt
g'\right)  ,\left(  h\lt g' \right)  h' \right).
\end{equation}  
The source and target maps of the ``matched pair Lie groupoid'' $\C{G} \bowtie \C{H}$ are given by
\begin{align*}
& \a: \C{G} \bowtie \C{H} \longrightarrow B, \qquad (g,h)  \mapsto \a(g),
\\
& \b: \C{G} \bowtie \C{H} \longrightarrow B, \qquad (g,h)  \mapsto \b(h),
\end{align*}
respectively. The object inclusion map of the matched pair Lie groupoid is defined in terms of those on  $\C{G}$ and $\C{H}$ as
\[
\ve: B \longrightarrow \C{G} \bowtie \C{H}, \qquad  b \mapsto (\widetilde{b},\widetilde{b}).
\]
The relation between the matched pair Lie groupoid $\C{G} \bowtie \C{H}$ and the individual Lie groupoids $\C{G}$ and $\C{H}$ is given in \cite[Thm. 2.10]{Mack92} that we record below. 
\begin{proposition}\label{prop-grpd-embeddings}
A pair $(\C{G},\C{H})$ of groupoids is a matched pair of groupoids if and only if the manifold $\C{G} \ast \C{H}$ has the structure of a Lie groupoid, such that
\begin{itemize}
\item[(i)] the maps $\C{G} \to \C{G} \ast \C{H}$ given by $g\mapsto (g, \widetilde{\b(g)})$, and $\C{H} \to \C{G} \ast \C{H}$ given by $h \mapsto (\widetilde{\a(h)},h)$ are morphisms of Lie groupoids, and
\item[(ii)] the multiplication $((g, \widetilde{\b(g)}), (\widetilde{\a(h)},h)) \mapsto (g, h) \in \C{G} \ast \C{H}$ is a diffeomorphism.
\end{itemize}
\end{proposition}

Picturing an element $(g,h) \in \C{G} \bowtie \C{H}$ by
\begin{equation} \label{corner}
\xymatrix{&&\b(h)\\ \\ \a(g) \ar [rr]_{g \quad}&& \b(g)=\a(h), \ar[uu]^{h}
}
\end{equation}  
the partial multiplication \eqref{mpmul} on the matched pair groupoid may be illustrated as 
\begin{equation}
\xymatrix{&&&&&\b(h')\\ \\ &&\b(h)=\a(g')
\ar [rrr]_{g'}&&&t(g')=\b(h\lt g')\ar[uu]^{h'}\\ \\ \a(g)\ar[rr]_{g\qquad\qquad}&&\ar[uu]^{h}
\b(g)=\a(h)=\a(h\rt g')\ar [rrr]_{h\rt g'}&&& \b(h\rt g')=\a(h\lt g') \ar[uu]^{h\lt g'}.
}
\end{equation}  

The following remark concerns the actions on the identity elements.
\begin{remark}
Let $(\C{G},\C{H})$ be a matched pair of Lie groupoids. Given $g\in \C{G}$ with $\a(g)=b_1$ and $\b(g)=b_2$, and $h\in \C{H}$ with $\a(h) = b_3$ and $\b(h)=b_4$, we see at once that
\begin{equation}\label{act}
h \rt \widetilde{b_4} = \widetilde{b_3} \in \C{G}, \qquad \widetilde{b_1} \lt g = \widetilde{b_2} \in \C{H},
\end{equation}
and that,
\begin{equation}\label{act2}
\widetilde{b_1} \rt g = g, \qquad h \lt \widetilde{b_4} = h.
\end{equation}
\end{remark}

\subsubsection{Matched pair decomposition of the trivial groupoid} \label{ex-triv-grpd-as-matched-pair}~

\noindent
Given the action groupoid $\C{G}=M\times G$ of Example \ref{ex-action-groupoid}, and the coarse groupoid $\C{H}=M\times M$ of Example \ref{ex-coarse-groupoid}, let us consider the set 
\[
\C{H}\ast\C{G}=(M \times M) \ast (M \times G)=\big\{((m', m;m, g))\in (M \times M) \times (M \times G)\big\}
\]
of composable elements. The left action
\begin{equation}\label{left-act}
\rt: (M \times M) \ast (M \times G) \lra (M \times G), \qquad (m', m) \rt (m, g) := (m', g)
\end{equation}
of the action groupoid $\C{G}=M\times G$ on the coarse groupoid $\C{H}=M\times M$, and the right action
\begin{equation}\label{right-act}
\lt: (M \times M) \ast (M \times G) \lra (M \times M), \qquad (m', m) \lt (m, g) := (m'g, mg)
\end{equation}
of the action groupoid $\C{G}=M\times G$ on the coarse groupoid $\C{H}=M\times M$ satisfies the conditions (i)-(ix) of the previous subsection. Thus, the set 
\begin{equation}
\C{G}\ast\C{H}=(M \times G)\ast (M\times M)=\big\{(m,g;mg,m')\in (M \times G)\times (M\times M)\big\}
\end{equation}
of composable elements form a Lie groupoid $\C{G}\bowtie\C{H}=(M \times G)\bowtie (M\times M)$ over the base manifold $M$. The source, target and object inclusion maps are computed to be
\begin{eqnarray}
\a &:&(M \times G)\bowtie (M\times M) \longrightarrow M, \qquad (m,g;mg,m')  \mapsto m,
\\
\b &:& (M \times G)\bowtie (M\times M) \longrightarrow M, \qquad (m,g;mg,m')  \mapsto m',
\\
\ve &:& M \longrightarrow (M \times G)\bowtie (M\times M), \qquad  m \mapsto (m,e;m,m).
\end{eqnarray}
In order to proceed to the partial multiplication, we consider the product space
\begin{eqnarray*}
&&\left((M \times G)\bowtie (M\times M)\right) \ast \left((M \times G)\bowtie (M\times M)\right):=\\&&\big\{(m,g;mg,m'),(m',h;m'h,n):m',m,n\in M \text{ and } g,h\in G \big\}.
\end{eqnarray*}
The partial multiplication, given by \eqref{mpmul}, then appears as
\begin{eqnarray*}
(m,g;mg,m')\ast(m',h;m'h,n)&=& \left(  (m,g) \left(  (mg,m')\rt
(m',h)\right);\left(  (mg,m')\lt (m',h) \right)  (m'h,n) \right)
\\
&=& (  (m,g) (mg,h);(mgh,m'h)(m'h,n) )
\\
&=& (m, gh;mgh,n).
\end{eqnarray*}
Accordingly, the inversion is computed to be
\[
(m,g;mg,n)^{-1} = (n,g^{-1}; ng^{-1},m).
\]

The matched pair Lie groupoid $(M \times G)\bowtie (M\times M)$ is identified with the trivial Lie groupoid $M\times G \times M$ via
\begin{equation}\label{map-Phi}
\Phi: M\times G \times M  \longrightarrow (M \times G)\bowtie (M\times M), \qquad (m,g,n) \mapsto (m,g;mg,n),
\end{equation}
see for instance, \cite{Mokr97}.

Let us finally note that the map \eqref{map-Phi} that gives the matched pair decomposition of the trivial groupoid is differentiated to a Lie algebroid morphism 
\begin{align}\label{map-APhi}
\begin{split}
& \C{A}_m\Phi:\C{A}_m(M\times G\times M) \lra \C{A}_m((M\times G)\bowtie (M\times M) ), \\
& (\t_m,\xi,Y) \mapsto (\t_m,\xi;\xi^\dagger(m),Y).
\end{split}
\end{align}
Indeed, for a curve $(m,e_t,m_t) \in M\times G \times M$ with $e_0=e$, $m_0 = m$, $\dot{e}_0 = \xi\in\G{g}$, and $\dot{m}_0 = Y \in T_mM$, recalling \eqref{dagger} we compute
\begin{align*}
(\C{A}_m\Phi)(\t_m,\xi,Y) = \dt\, \Phi(m,e_t,m_t) =  \dt\, (m,e_t;me_t,m_t) = (\t_m,\xi;\xi^\dagger(m),Y).
\end{align*}

\subsection{Matched pairs of Lie algebroids}~

\noindent 
Let us begin with a brief discussion on the representation of a Lie algebroid on a vector bundle from \cite{HiggMack90,Mokr97,Lu97}. 

Let $(\C{A},\tau,M,a,[\bullet,\bullet])$ be a Lie algebroid, and let $(E,\pi,M)$ be a vector bundle over the same base manifold $M$. A left representation of $(\C{A},\tau,M)$ to $(E,\pi,M)$ is a bilinear map
\[
\rho:\Gamma(\C{A})\times \Gamma(E) \lra \Gamma(E), \qquad (X,s)\mapsto \rho_X(s)=\rho(X,s)
\] 
such that
\begin{itemize}
\item[(i)] $\rho_{fX}(Y) = f\rho_X(Y)$,
\item[(ii)] $\rho_X(fY) = f\rho_X(Y) + (a(X)f)Y$,
\item[(iii)] $\rho_{[X,\widetilde{X} ]}(Y) = \rho_X(\rho_{\widetilde{X}}(Y)) - \rho_{\widetilde{X}}(\rho_X(Y))$,
\end{itemize}
for any $X,\widetilde{X} \in \Gamma(\C{A})$, any $Y\in \Gamma(E)$, and any $f \in C^\infty(M)$. A right representation of a Lie algebroid on a vector bundle is defined similarly. 

Two Lie algebroids with the same base manifold form a matched pair of Lie algebroids if the direct sum of the total spaces of the Lie algebroids has a Lie algebroid structure on the same base such that the individual Lie algebroids are Lie subalgebroids of the direct sum, \cite{Mokr97}. More precisely, let $(\C{A},\tau,M,a,[\bullet,\bullet])$ and $(\C{B},\kappa,M,b,[\bullet,\bullet])$ be two Lie algebroids over the same base $M$ with mutual representations 
\[
\rho: \Gamma(\C{B})\times \Gamma(\C{A}) \to \Gamma(\C{A}), \qquad \rho': \Gamma(\C{A})\times \Gamma(\C{B}) \to \Gamma(\C{B}),
\] 
satisfying
\begin{itemize}
\item[(i)] $\rho_Y[X,\widetilde{X}] = [\rho_Y(X),\widetilde{X}] + [X,\rho_Y(\widetilde{X})] - \rho_{\rho'_X(Y)}(\widetilde{X}) + \rho_{\rho'_{\widetilde{X}}(Y)}(X)$,
\item[(ii)] $\rho'_X[Y,\widetilde{Y}] = [\rho_X(Y),\widetilde{Y}] + [Y,\rho_X(\widetilde{Y})] - \rho'_{\rho_Y(X)}(\widetilde{X}) + \rho'_{\rho_{\widetilde{Y}}(X)}(Y)$,
\item[(iii)] $[b(Y),a(X)]=a(\rho_Y(X)) - b(\rho'_X(Y))$,
\end{itemize}
for any $X,\widetilde{X} \in \Gamma(\C{A})$, and any $Y,\widetilde{Y} \in \Gamma(\C{B})$. Then, the direct sum vector bundle $\C{A}\bowtie \C{B}:=\C{A}\oplus \C{B}$ has the structure of a Lie algebroid by the bracket given by
\begin{equation}\label{Lie-algoid-actions}
[Y,X] = \rho(Y,X) - \rho'(X,Y)
\end{equation}
 for any $X \in \Gamma(\C{A})$, and any $Y\in \Gamma(\C{B})$. The pair $(\C{A},\C{B})$ of Lie algebroids is called a ``matched pair of Lie algebroids'', whereas the vector bundle $\C{A}\bowtie \C{B}$ is called a ``matched pair Lie algebroid''. 
 
%


\subsection{Lie algebroid actions induced from Lie groupoid actions} ~

\noindent
We devote the present subsection on the infinitesimal versions of the mutual actions of a matched pair $(\C{G},\C{H})$ of Lie groupoids, over a base manifold $B$.

To this end, let $h\in \C{H}$, and  let $x_t\in \C{G}$ be a curve so that $\b_H(h)=\a_G(x_t)=b \in B$. Since the curve has a constant source, its derivative at $t=0$ qualifies as a Lie algebroid element, which we denote by $X\in \C{A}_b\C{G}$. Now, using the left action, we define the curve $\overline{x}_t:=h\rt x_t\in \C{G}$, whose source is constant as well; $\a(h\rt x(t))=\a(h)=c\in B$. As such, the time derivative (at $t=0$) of the latter also lies in the kernel of the tangent mapping of the source map, and we arrive at the element 
\begin{equation} \label{honX}
h\rt X:=\dt (h\rt x_t) \in \C{A}_c\C{G}.
\end{equation}
On the other hand, the right action allows us to define the curve $h \lt x_t\in\C{H}$ which passes through $h\in \C{H}$ at $t=0$ by \eqref{act2}. We note that the source of the curve $h \lt x_t\in\C{H}$ is not necessarily constant, and that its time derivative  
\begin{equation} \label{Xonh}
X^{\dagger}(h)=h\lt X:=\dt (h\lt x_t) \in T_h\C{H}
\end{equation} 
may not be in $\C{A}_c\C{H}$.

Let, next, $y_t\in\C{H}$ be a curve with constant source, say $\a_\C{H}(y_t) = c = \a_\C{G}(g) \in B$ for some $g\in \C{G}$, and thus generating a Lie algebroid element $Y\in \C{A}_c\C{H}$. Then the left infinitesimal action of $\C{AH}$ on $\C{G}$ is defined as 
\begin{equation} \label{Yong}
Y^{\dagger}(g):= \dt y^{-1}_t\rt g \in T_g\C{G}.
\end{equation}
Similarly, the right action of $\C{G}$ on $\C{H}$ can be lifted to the right action of $\C{G}$ on $\C{AH}$ as follows. Starting with the curve $y_t \in \C{H}$, we define the curve $(y^{-1}_t\lt g)^{-1}\in \C{H}$. The source of the latter curve is constant, hence, at $t=0$ the curve is tangent to an element in $\C{A}_c\C{H}$, which is given by 
\begin{equation} \label{gonY}
Y\lt g:= \dt (y^{-1}_t\lt g)^{-1} \in \C{A}_c\C{H}.
\end{equation}

\subsection{Lie algebroid of a matched pair Lie groupoid}\label{subsect-Lie-algbrd-of-matched-Lie-grpd} ~

\noindent
In Subsection \ref{susubsect-Lie-algbrd-of-a-lie-grpd}, we have seen how a Lie algebroid is associated to a Lie groupoid. In the present subsection, we shall revisit this construction for a matched Lie groupoid, say $\C{G}\bowtie\C{H}$ to arrive at the Lie algebroid $\C{A}(\C{G}\bowtie\C{H})$. We shall conclude with the explicit isomorphism between $\C{A}(\C{G}\bowtie\C{H})$ and the matched pair Lie algebroid $\C{A}\C{G}\bowtie \C{A}\C{H}$.

To begin with, let $\C{G}\bowtie\C{H}$ be a matched pair of Lie groupoids, $X \in \C{A}_b\C{G}$, and $Y\in \C{A}_b\C{H}$; that is, let there be curves $x_t \in \C{G}$ and $y_t\in \C{H}$ so that $\a(x_t) = b = \a(y_t)$, and that $x_0 = \ve_G(b) \in \C{G}$ with $y_0 = \ve_H(b) \in \C{H}$. In view of Proposition \ref{prop-grpd-embeddings}, the (groupoid) embeddings
\begin{eqnarray}
\C{G} &\lra & \C{G}\bowtie \C{H}, \qquad x_t\mapsto (x_t,(\ve_\C{H}\circ \b)(x_t))
\\
\C{H} &\lra & \C{G}\bowtie \C{H}, \qquad y_t\mapsto (\ve_\C{G}(b) ,y_t).
\end{eqnarray}
induce morphisms
\begin{eqnarray}
\C{A}_b\C{G} &\lra & \C{A}_b(\C{G}\bowtie \C{H}), \qquad X\mapsto (X,T(\ve_\C{H}\circ \b)(X))
\\
\C{A}_b\C{H} &\lra & \C{A}_b(\C{G}\bowtie \C{H}), \qquad Y\mapsto (\t_{\ve_\C{G}(b)} ,Y)
\end{eqnarray}
of Lie algebroids. Furthermore, we have the following proposition, see \cite[Prop. 5.1]{Mokr97}.

\begin{proposition} \label{prop-matched-algoid-isom}
Given two Lie groupoids $\C{G}$ and $\C{H}$ over the same base $B$, the map 
\begin{equation}\label{eqn-AG+AH-to-AGH}
\C{A}_b\C{G}\oplus \C{A}_b\C{H} \to \C{A}_b(\C{G}\bowtie \C{H}), \qquad (X,Y) \mapsto (X,T(\ve_\C{H}\circ \b)(X)+Y).
\end{equation}
is an isomorphism. 
\end{proposition}
\begin{proof}
Let $(x_t,z_t) \in \C{G}\bowtie\C{H}$ be a curve such that $\a(x_t)=b \in B$, and that $\b(x_t) = \a(z_t)$. Let also $\dot{x}_0 = X \in \C{A}_b\C{G}$, and $\dot{z}_0 = Z \in T_{\ve_\C{H}(b)}\C{H}$. Multiplying both sides of
\[
((\ve_\C{G}\circ\b)(z_t),z_t^{-1})(x_t^{-1},\ve_\C{H}(b))(x_t,z_t) = ((\ve_\C{G}\circ\b)(z_t),(\ve_\C{H}\circ\b)(z_t)).
\]
by $((\ve_\C{G}\circ\a)(z_t),z_t) \in \C{G}\bowtie\C{H}$, from the left, we obtain
\[
(x_t^{-1},\ve_\C{H}(b))(x_t,z_t) = ((\ve_\C{G}\circ\a)(z_t),z_t).
\]
Differentiating the latter equality, while keeping the product rule in mind, we arrive at
\[
\left(\dt\,(x_t^{-1},\ve_\C{H}(b))\right) + \left(\dt\,(x_t,z_t)\right) = \left(\dt\,((\ve_\C{G}\circ\a)(z_t),z_t)\right),
\]
that is,
\begin{equation}\label{eqn-x-z}
\left(\dt\,(x_t,z_t)\right) = \left(\dt\,((\ve_\C{G}\circ\a)(z_t),z_t)\right) - \left(\dt\,(x_t^{-1},\ve_\C{H}(b))\right).
\end{equation}
On the other hand, differentiating
\[
(x_t^{-1},\ve_\C{H}(b))(x_t,(\ve_\C{H}\circ \b)(x_t)) = ((\ve_\C{G}\circ \b)(x_t),(\ve_\C{H}\circ \b)(x_t)),
\]
we obtain
\[
\left(\dt\,(x_t^{-1},\ve_\C{H}(b))\right) + \left(\dt\,(x_t,(\ve_\C{H}\circ \b)(x_t))\right) = \left(\dt\,((\ve_\C{G}\circ \b)(x_t),(\ve_\C{H}\circ \b)(x_t))\right),
\]
that is,
\begin{equation}\label{eqn-x-inv-x}
\left(\dt\,(x_t^{-1},\ve_\C{H}(b))\right)= \left(\dt\,((\ve_\C{G}\circ \b)(x_t),(\ve_\C{H}\circ \b)(x_t))\right) - \left(\dt\,(x_t,(\ve_\C{H}\circ \b)(x_t))\right) .
\end{equation}
Now, \eqref{eqn-x-z} and \eqref{eqn-x-inv-x} together imply
\begin{align*}
& \left(\dt\,(x_t,z_t)\right) = \\
& \left(\dt\,(x_t,(\ve_\C{H}\circ \b)(x_t))\right) +\left(\dt\,((\ve_\C{G}\circ\a)(z_t),z_t)\right) - \left(\dt\,((\ve_\C{G}\circ \b)(x_t),(\ve_\C{H}\circ \b)(x_t))\right) = \\
& \left(\dt\,(x_t,(\ve_\C{H}\circ \b)(x_t))\right) +\left(\dt\,((\ve_\C{G}\circ \b)(x_t),z_t)\right) - \left(\dt\,((\ve_\C{G}\circ \b)(x_t),(\ve_\C{H}\circ \b)(x_t))\right),
\end{align*}
where, on the second equality we used the fact that $\b(x_t) =\a(z_t)$. Hence, we see at once that
\begin{align*}
& \left(\dt\,(x_t,z_t)\right) - \left(\dt\,(x_t,(\ve_\C{H}\circ \b)(x_t))\right) = \\
& \left(\dt\,((\ve_\C{G}\circ \b)(x_t),z_t)\right) - \left(\dt\,((\ve_\C{G}\circ \b)(x_t),(\ve_\C{H}\circ \b)(x_t))\right) = \\
& (\t_{\ve_\C{G}(b)}, Z - T(\ve_\C{H}\circ \b)(X)) \in \C{A}_b(\C{G}\bowtie \C{H}).
\end{align*}
As a result, $Z - T(\ve_H\circ \b)(X) \in T_{\ve_\C{H}(b)}H$ is the derivative, at $t=0$, of a curve $y_t \in \C{H}$ with $\a(y_t) = b$; that is, $Z - T(\ve_H\circ \b)(X) \in \C{A}_b\C{H}$. In other words, the mapping
\[
\C{A}(\C{G}\bowtie \C{H}) \lra \C{A}\C{G} \oplus \C{A}\C{H}, \qquad (X,Z) \mapsto (X,Z - T(\ve_\C{H}\circ \b)(X))
\]
is well-defined, and makes the inverse of \eqref{eqn-AG+AH-to-AGH}.
\end{proof}

The next example illustrates the isomorphism\eqref{eqn-AG+AH-to-AGH} for the matched pair $(M\times G,M\times M)$ of the action Lie groupoid of Example \ref{ex-action-groupoid} and the coarse groupoid $M\times M$ of Example \ref{ex-coarse-groupoid}.

\begin{example}
Let $(\t_m,\xi;\xi^\dagger(m),Y) \in \C{A}_m((M\times G)\bowtie (M\times M))$ be a generic element obtained by the differentiation at $t=0$ of a curve $(m,g_t;mg_t,n_t) \in (M\times G)\bowtie (M\times M)$ of constant source, with $g_0=e\in G$, $n_0=m$, $\dot{g}_0 = \xi \in \G{g}$ and $\dot{n}_0=Y\in T_mM$. Similarly, the inclusion
\[
M\times G \ni (m,g_t)\mapsto (m,g_t;(\ve_{M\times M}\circ\b)(m,g_t)) = (m,g_t;mg_t,mg_t))  \in (M\times G)\bowtie (M\times M)
\]
yields
\[
\C{A}_m(M\times G)\ni (\t_m,\xi) \mapsto (\t_m,\xi;\xi^\dagger(m),\xi^\dagger(m))\in\C{A}_m((M\times G)\bowtie (M\times M)),
\]
while 
\[
M\times M \ni (m,n_t)\mapsto (m,e;m,n_t)) \in (M\times G)\bowtie (M\times M)
\]
leads to
\[
\C{A}_m(M\times M)\ni (\t_m,X) \mapsto (\t_m,\t;\t_m,X)\in\C{A}_m((M\times G)\bowtie (M\times M)).
\]
As a result, we see that \eqref{eqn-AG+AH-to-AGH} takes the form
\begin{align}\label{eqn-iso-AMG-AMM-to-AMGMM}
\begin{split}
& \C{A}_m(M\times G) \oplus \C{A}_m(M\times M)\ni(\t_m,\xi)\oplus (\t_m,X) \mapsto \\
&\hspace{3cm} (\t_m,\xi;\xi^\dagger(m),X+\xi^\dagger(m)) \in \C{A}_m((M\times G)\bowtie (M\times M)).
\end{split}
\end{align}
\end{example}

\subsection{Left and right invariant vector fields on matched pairs of Lie groupoids}
~

\noindent
In this subsection we shall determine the nature of the left invariant (resp. the right invariant) vector fields on a matched pair Lie groupoid. 

\begin{proposition}
Let $\C{G}\bowtie \C{H}$ be a matched pair Lie groupoid over a base manifold $B$. Then, the left invariant vector field corresponding to $U \in \C{A}_b(\C{G}\bowtie \C{H})$ is given by
\begin{equation}\label{eqn-matched-LIVF}
\overleftarrow{U}(g,h) = (\overleftarrow{h\rt X}(g), X^\dagger(h)+\overleftarrow{Y}(h)),
\end{equation}
where $g\in \C{G}$, $h\in \C{H}$, so that $\b(h) = b$, and $X \in \C{A}_b\C{G}$, $Y \in \C{A}_b\C{H}$.
\end{proposition}

\begin{proof}
In view of the isomorphism \ref{prop-matched-algoid-isom}, any $U \in A_b(\C{G}\bowtie \C{H})$ may be written as
\[
U = (X,T(\ve\circ \b)(X)) + (\t_{\ve_G(b)},Y)
\]
for some $X \in A_b\C{G}$, and $Y \in A_b\C{H}$. As such,
\[
\overleftarrow{U}(g,h) = \overleftarrow{(X,T(\ve\circ \b)(X))} (g,h)+ \overleftarrow{(\t_{\ve_G(b)},Y)}(g,h).
\]
More precisely, assuming $\a(x_t) = \b(h)$, we have
\begin{align*}
& \overleftarrow{(X,T(\ve\circ \b)(X))}(g,h) = \dt\,(g,h)(x_t,(\ve_H\circ \b)(x_t)) = \\
& \dt\, (g(h\rt x_t), (h\lt x_t)(\ve_H\circ \b)(x_t)) = \\
& \dt\, (g(h\rt x_t), (h\lt x_t)) = (\overleftarrow{h\rt X}(g), h\lt X) = (\overleftarrow{h\rt X}(g), X^\dagger(h)).
\end{align*}
Similarly, assuming $\b(y_t) = \a(g)$ and $\b(h) = b = \a(y_t)$, we have
\begin{align*}
& \overleftarrow{(\t_{\ve_G(b)},Y)}(g,h) = \dt\, (g,h)(\ve_G(b),y_t) =\dt\, (g(h\rt \ve_G(b)),(h\lt \ve_G(b))y_t) = \\
& \dt\, (g\ve_G(\a(h)), hy_t) = (\t_g, \overleftarrow{Y}(h)).
\end{align*}
The result follows.
\end{proof}

The right analogue is given by the following proposition.

\begin{proposition}
Let $\C{G}\bowtie \C{H}$ be a matched pair of Lie groupoids over a base manifold $B$. Then, the right invariant vector field corresponding to $U \in A_b(\C{G}\bowtie \C{H})$ is given by
\begin{equation}\label{eqn-matched-RIVF}
\overrightarrow{U}(g,h) = (\overrightarrow{X}(g)-Y^\dagger(g),\overrightarrow{Y\lt g}(h)),
\end{equation}
where $g\in \C{G}$, so that $\a(g) = b$, $h\in \C{H}$, $X \in \C{A}_b\C{G}$, and $Y \in \C{A}_b\C{H}$.
\end{proposition}

\begin{proof}
This time we have for any $U \in A_b(\C{G}\bowtie \C{H})$ that
\[
\overrightarrow{U} (g,h)= \overrightarrow{(X,T(\ve\circ \b)(X))} (g,h)+ \overrightarrow{(\t_{\ve_G(b)},Y)}(g,h).
\]
Accordingly, assuming $\a(x_t) = b =\a(g)$,
\begin{align*}
& \overrightarrow{(X,T(\ve\circ \b)(X))}(g,h) = -\dt\, (x_t,(\ve_H\circ \b)(x_t))^{-1}(g,h) = \\
& -\dt\, ((\ve_H\circ \b)(x_t)^{-1}\rt x_t^{-1},(\ve_H\circ \b)(x_t)^{-1}\lt x_t^{-1})(g,h) = \\
& -\dt\, (x_t^{-1},\ve_H(b))(g,h) =-\dt\, (x_t^{-1}(\ve_H(b) \rt g),(\ve_H(b) \lt g)h) = \\
& -\dt\, (x_t^{-1}g,h) = (\overrightarrow{X}(g),\t_h),
\end{align*}
and
\begin{align*}
& \overrightarrow{(\t_{\ve_G(b)},Y)}(g,h) = -\dt\, (\ve_G(b),y_t)^{-1}(g,h) = \\
& -\dt\, (y_t^{-1}\rt\ve_G(b)^{-1},y_t^{-1}\lt\ve_G(b)^{-1})(g,h) = \\
& -\dt\, ((\ve_G\circ \b)(y_t),y_t^{-1})(g,h) = \\
& -\dt\, ((\ve_G\circ \b)(y_t)(y_t^{-1} \rt g),(y_t^{-1} \lt g)h) = \\
& -\dt\, (y_t^{-1} \rt g,(y_t^{-1} \lt g)h) = (-Y^\dagger(g),\overrightarrow{Y\lt g}(h)).
\end{align*}
The result follows.
\end{proof}

\begin{example}\label{ex-left-inv-right-inv-vf-MG-MM}
Let us now derive the left (resp. right) invariant vector fields on $(M\times G) \bowtie (M\times M)$ in view of the isomorphism \eqref{eqn-iso-AMG-AMM-to-AMGMM}.

Given $(m,g;mg,n) \in (M\times G) \bowtie (M\times M)$ and $(\t_n,\xi;\xi^\dagger(n),Y) \in A_n((M\times G) \bowtie (M\times M))$, we have
\begin{align*}
& \overleftarrow{(\t_n,\xi;\xi^\dagger(n),Y)}(m,g;mg,n)) = \\
& \dt\,(m,g;mg,n)(n,e_t;ne_t,n_t) = \\
& \dt\, \Big((m,g)\big( (mg,n)\rt (n,e_t)\big);\big( (mg,n)\lt (n,e_t)\big)(ne_t,n_t)\Big) = \\
& \dt\, \Big((m,g)(mg,e_t);(mge_t,ne_t)(ne_t,n_t)\Big)  =  \dt\, \Big(m,ge_t;mge_t,n_t\Big)  = \\
&  (\t_m,\overleftarrow{\xi}(g); \xi^\dagger(mg),Y).
\end{align*}
Similarly, for $(m,g;mg,n) \in (M\times G) \bowtie (M\times M)$ and $(\t_m,\xi;\xi^\dagger(m),Y) \in A_m((M\times G) \bowtie (M\times M))$,
\begin{align*}
& \overrightarrow{(\t_m,\xi;\xi^\dagger(m),Y) }(m,g;mg,n) = \\
& -\dt\,(m,e_t;me_t,n_t)^{-1}(m,g;mg,n) = \\
& - \dt\, \Big((me_t,n_t)^{-1}\rt (m,e_t)^{-1};(me_t,n_t)^{-1}\lt (m,e_t)^{-1}\Big) (m,g;mg,n)= \\
& - \dt\, \Big((n_t,me_t)\rt (me_t,e_t^{-1});(n_t,me_t)\lt (me_t,e_t^{-1})\Big)(m,g;mg,n) = \\
& - \dt\, \Big(n_t,e_t^{-1};n_te_t^{-1},m\Big) (m,g;mg,n)= \\
& - \dt\, \Big((n_t,e_t^{-1})\big((n_te_t^{-1},m)\rt (m,g)\big);\big((n_te_t^{-1},m)\lt (m,g)\big)(mg,n)\Big)= \\
& - \dt\, \Big((n_t,e_t^{-1})(n_te_t^{-1},g);(n_te_t^{-1}g,mg)(mg,n)\Big)= \\
&- \dt\, \Big(n_t,e_t^{-1}g;n_te_t^{-1}g,n\Big)= (-Y,\overrightarrow{\xi}(g);(Y-\xi^\dagger(m))\lt g,\t_n).
\end{align*}
On the other hand, in view of \eqref{eqn-iso-AMG-AMM-to-AMGMM}, for any $(m,g;mg,n) \in (M\times G)\bowtie (M\times M)$, and
\[
(\t_n,\xi;\xi^\dagger(n),Y) = (\t_n,\xi;\xi^\dagger(n),\xi^\dagger(n)) + (\t_n,\t;\t_n,X) \in \C{A}_n((M\times G)\bowtie (M\times M)),
\]
we have
\begin{align*}
& \overleftarrow{(\t_n,\xi;\xi^\dagger(n),Y)} (m,g;mg,n)= \\
& \overleftarrow{(\t_n,\xi;\xi^\dagger(n),\xi^\dagger(n))}(m,g;mg,n) + \overleftarrow{(\t_n,\t;\t_n,X)} (m,g;mg,n)= \\
& (\t_m,\overleftarrow{\xi}(g);\xi^\dagger(mg),\xi^\dagger(n)) + (\t_m,\t_g;\t_{mg},X) = \\
& (\t_m,\overleftarrow{\xi}(g);\xi^\dagger(mg),\xi^\dagger(n)+X),
\end{align*}
and for any 
\[
(\t_m,\xi;\xi^\dagger(m),Y) = (\t_m,\xi;\xi^\dagger(m),\xi^\dagger(m)) + (\t_m,\t;\t_m,Z) \in \C{A}_m((M\times G)\bowtie (M\times M)),
\]
we obtain
\begin{align*}
& \overrightarrow{(\t_m,\xi;\xi^\dagger(m),Y)}(m,g;mg,n) = \\
& \overrightarrow{(\t_m,\xi;\xi^\dagger(m),\xi^\dagger(m))} (m,g;mg,n)+ \overrightarrow{(\t_m,\t;\t_m,Z)} (m,g;mg,n)=\\
& (-\xi^\dagger(m),\overrightarrow{\xi}(g);\t_{mg},\t_n) + (-Z,\t_g;-Z\lt g,\t_n) = \\
& (-\xi^\dagger(m)-Z,\overrightarrow{\xi}(g);-Z\lt g,\t_n).
\end{align*}
\end{example}

\begin{remark}
We can relate the above calculations to the left (resp. right) invariant vector fields on the trivial groupoids as follows. The (groupoid) isomorphism \eqref{map-Phi} induces the isomorphism \eqref{map-APhi} on the level of Lie algebroids. Hence, it induces an isomorphism on the level of left (resp. right) invariant vector fields. Indeed,
\begin{align*}
& \overleftarrow{(\t_n,\xi;\xi^\dagger(n),Y)} (m,g;mg,n)= \overleftarrow{\C{A}_n\Phi(\t_n,\xi,Y)} (m,g;mg,n) \\
& \dt\,(m,g;mg,n)(n,e_t;ne_t,n_t) = \\
& \dt\,\Phi(m,g,n)\Phi(n,e_t,n_t) = \dt\, \Phi\big((m,g,n)(n,e_t,n_t)\big) = \\
& T_{(m,g,n)}\Phi\left(\overleftarrow{(\t_n,\xi,Y)}(m,g,n)\right).
\end{align*}
Similarly,
\begin{align*}
& \overrightarrow{(\t_m,\xi;\xi^\dagger(m),Y)} (m,g;mg,n)= \overrightarrow{\C{A}_m\Phi(\t_m,\xi,Y)} (m,g;mg,n) \\
& -\dt\,(n,e_t;ne_t,n_t)^{-1}(m,g;mg,n) = \\
& \dt\,\Phi\big((n,e_t,n_t) ^{-1}\big)\Phi(m,g,n)= \dt\, \Phi\big((n,e_t,n_t) ^{-1}(m,g,n)\big) = \\
& T_{(m,g,n)}\Phi\left(\overrightarrow{(\t_m,\xi,Y)}(m,g,n)\right).
\end{align*}
\end{remark}

\section{Discrete dynamics on matched pairs}\label{sect-disc-dynm-matched-pair}

\subsection{Discrete Euler-Lagrange equations}\label{subsect-disc-Euler-Lagrange-eqns}~

We shall recall briefly the discrete Euler-Lagrange equations from \cite[Subsect. 4.1]{MarrDiegMart06}. 
Let $\mathcal{G}$ be a Lie groupoid, and $\C{A}\mathcal{G}$ be its associated Lie algebroid. For any fixed $g\in \mathcal{G}$ and $N\geq 1$, $\mathcal{G}^{N}$ being the $N$ times cartesian product of $\mathcal{G}$, the set
\[
\mathcal{C}_{g}^{N}=\left\{ \left( g_{1},...,g_{N}\right) \in \mathcal{G}^{N} \mid \left(
g_{k},g_{k+1}\right) \in \mathcal{G}\ast \C{G}, \,1\leq k \leq N-1,\, g_{1}...g_{N}=g\right\} 
\]
is called the set of admissible sequences with values in $\C{G}$.

On the other hand, a discrete Lagrangian is defined as a function $L:\C{G} \to \B{R}$, and the discrete action
sum associated to it is given by
\begin{equation}\label{disct-action-sum}
\C{S}L:\mathcal{C}_{g}^{N}\rightarrow 
\mathbb{R}, \qquad \left( g_{1},...,g_{N}\right) \mapsto \sum_{k=1}^{N}L\left(
g_{k}\right) .
\end{equation}
Now the discrete Hamilton's principle may be recalled, from \cite{Wein96}, as follows. Given $g\in \mathcal{G}$ and $N\geq 1$, an admissible sequence 
$\left( g_{1},...,g_{N}\right) $ is a solution of the Lagrangian system
if and only if $\left( g_{1},...,g_{N}\right) \in \C{C}_g^N$ is a critical
point of \eqref{disct-action-sum}. So, along the lines of \cite{MarrDiegMart06}, one arrives at the discrete Euler-Lagrange equations 
\[
\sum_{k=1}^{N-1}[\overleftarrow{X}_{k}\left( g_{k}\right) (L)-
\overrightarrow{X}_{k}\left( g_{k+1}\right)
(L)]=0,
\]%
for any $X_{k} \in \Gamma(\C{AG})$. In particular, for $N=2$, the discrete Euler-Lagrange equations are given by
\begin{equation}\label{eqn-disct-EL}
\overleftarrow{X}\left( g_{1}\right) (L)-\overrightarrow{X}\left(
g_{2}\right) (L)=0
\end{equation}
for every section $X \in \Gamma(A\mathcal{G})$. 

We review below the examples discussed in \cite{MarrDiegMart06}. 

\begin{example}
Let $M\times M$ be the coarse (pair) groupoid of Example \ref{ex-coarse-groupoid}, whose left invariant vector fields (resp. the right invariant vector fields) correponding to the Lie algebroid of $M\times M$ were obtained in Example \ref{ex-pair-grpd-left-right-inv-v-f}.

Now, given a discrete Lagrangian $L:M\times M\to \B{R}$, the discrete Euler-Lagrange equations \eqref{eqn-disct-EL} takes the particular form 
\begin{eqnarray} \label{D-EL-1}
\overleftarrow{X}(x,y)(L)-\overrightarrow{X}(x,y)(L)=0.
\end{eqnarray}
In terms of the total derivatives on the product manifold $M\times M$, the discrete Euler-Lagrange equations may be rewritten as
\begin{eqnarray}\label{D-EL-2}
D_{2}L(x,y)+D_{1}L(y,z)=0,
\end{eqnarray}
see also \cite{MaWe01}.
\end{example}

\begin{example}
Let $G$ be a Lie group (with the Lie algebra $\G{g}$), considered as a groupoid over the identity $\{e\}$. We recall that the left invariant and the right invariant vector fields corresponding to the Lie algebroid $\G{g}$ of $G$ coincides with the left invariant and the right invariant vector fields on $G$. 
 
Now, given a discrete Lagrangian density $L:G\to \B{R}$ the discrete Euler-Lagrange equations are given by
\begin{equation}\label{discrete-EL-on-Lie-group}
\overleftarrow{\xi}(g_k) (L)-\overrightarrow{\xi}(g_{k+1}) (L)=0,
\end{equation}
or equivalently
\begin{align*}
& \langle dL(g_k), \overleftarrow{\xi} (g_k)\rangle  - \langle dL(g_{k+1}), \overrightarrow{\xi} (g_{k+1})\rangle = \\
& \langle dL(g_k),  T_e\ell_{g_k}(\xi) \rangle - \langle dL(g_{k+1}), T_er_{g_{k+1}}(\xi)\rangle  = \\
& \langle T^\ast_e\ell_{g_k}(dL(g_k)) - T^\ast_e r_{g_{k+1}}(dL(g_{k+1})),  \xi\rangle  = 0,
\end{align*}
for any $\xi\in\G{g}$, and any $g_k,g_{k+1} \in G$. As such, the discrete Euler-Lagrange equations may be written by
\begin{equation}\label{eqn-discrete-EL}
T^\ast_e\ell_{g_k}(dL(g_k)) - T^\ast_e r_{g_{k+1}}(dL(g_{k+1})) = 0
\end{equation}

Following \cite{MarrDiegMart06}, we set 
\[
\mu_{k}:=\left(r_{g_{k}}^{\ast }dL\right) \left( e\right).
\]
Then \eqref{eqn-disct-EL} takes the form
\begin{align*}
& T^\ast_e\ell_{g_k}(dL(g_k)) - T^\ast_e r_{g_{k+1}}(dL(g_{k+1})) =  T^\ast_e\ell_{g_k}(dL(g_k)) - \mu_{k+1} = \\
& T^\ast_e\ell_{g_k}T^\ast_e (r_{g_k}\circ r_{{g_k}^{-1}})(dL(g_k)) - \mu_{k+1} =  T^\ast_e\ell_{g_k}T^\ast_e r_{{g_k}^{-1}} T^\ast_er_{g_k} (dL(g_k)) - \mu_{k+1}  = \\
& T^\ast_e\ell_{g_k}T^\ast_e r_{{g_k}^{-1}} (\mu_k) - \mu_{k+1}  =  \Ad^\ast_{g_k^{-1}}(\mu_k) - \mu_{k+1} = 0.
\end{align*}
In other words,
\[
\mu _{k+1}=\Ad_{g^{-1}_{k}}^{\ast }(\mu _{k}),
\]
called the discrete Euler-Lagrange equations, see also \cite{BobeSuri99,MaPeSh99,MaPeSh2000}.
\end{example}

\begin{remark}
We note for the adjoint action $\Ad:G\times \G{g} \to \G{g}$, $(g,\xi) \mapsto \Ad_g(\xi) =: g\rt \xi$ that
\[
\langle \mu, \Ad_g(\xi) \rangle = \langle \mu, g \rt \xi \rangle = \langle \mu \overset{\ast}{\lt} g, \xi\rangle = \langle g^{-1} \overset{\ast}{\rt}\mu, \xi\rangle = \langle \Ad^\ast_{g^{-1}}(\mu), \xi\rangle.
\]
\end{remark}

\begin{example}
Let $M\times G$ be the action Lie groupoid of Example \ref{ex-action-groupoid}, with the left invariant and the right invariant vector fields as in Example \ref{ex-action-grpd-left-right-inv-v-f}.

Given a Lagrangian $L: M\times G\to \B{R}$, the sequence
$((m,g_{k}),(m g_{k},g_{k+1}))\in (M\times G)\ast(M\times G)$ is a solution of the discrete
Euler-Lagrange equations if
\[
\overleftarrow{(\t_m,\xi)}\left(m,g_k\right) \left( L\right) -%
\overrightarrow{(\t_m,\xi)}\left(m\cdot g_k,g_{k+1}\right) \left( L\right) =0,
\]
for any $\xi \in \G{g}$. Equivalently,
\[
\left( T_{e}\ell_{g_{k}}\right) \left( \xi \right) (L_m)-\left(
T_{e}r_{g_{k+1}}\right) \left( \xi \right) (L_{m g_{k}})+\xi^\dagger\left(
m  g_{k}\right) \left( L_{g_{k+1}}\right) =0,
\]%
where $L_m:G \to \B{R}$ is the map given by $L_m(g):=L(m,g)$, and similarly $L_g:M\to \B{R}$ is the one given by $L_g(m):= L(m,g)$.

Setting $\mu _{k}\left( m,g_{k}\right) =d(L_m\circ r_{g_{k}})\left(e\right)$ as above, the discrete Euler-Lagrange equations appear to be
\[
\mu _{k+1}\left( m\cdot g_{k},g_{k+1}\right) =\Ad_{g_{k}}^{\ast }\mu _{k}\left(
m,g_{k}\right) +d\left( L_{g_{k+1}}\circ \left( \left(m\cdot g_{k}\right) \cdot
\right) \right) \left( e\right) ,
\]
where $\left(m\cdot g_{k}\right) \cdot :G\rightarrow M$ is given by $\left( m\cdot g_{k}\right) \cdot (g):=m\cdot (g_{k}g)$.

If, in particular, $M$ is the orbit space of a representation of $G$ on $V$, then the corresponding equations were first obtained in \cite{BobeSuri99,BoSu99}, and they are called the discrete Euler-Poincar\'e equations.
\end{example}

\begin{example}
Let $M\times G \times M$ be the trivial groupoid of Example \ref{ex-trivial-groupoid}, whose left invariant and right invariant vector fields are obtained in Example \ref{ex-triv-grpd-left-right-inv-v-f}.  Accordingly, given a Lagrangian $L:M\times G \times M \to \B{R}$, together with $(m_k,g_k,n_k), (m_{k+1},g_{k+1},n_{k+1}) \in M\times G \times M$ with 
\[
\b(m_k,g_k,n_k) = n_k = m_{k+1} = \a(m_{k+1},g_{k+1},n_{k+1}),
\]
the discrete Euler-Lagrange equations are given by
\[
(\t_{m_k}, \overleftarrow{\xi}(g_k), X)(L) - (-X,\overrightarrow{\xi}(g_{k+1}),\t_{n_{k+1}})(L) = 0.
\]
Setting $\mu _{k}=T^*r_{g_k}d_2L(m_k,g_k,n_k)$, the discrete Euler-Lagrange equations appear as
\begin{equation}\label{discrete-tirivial-I}
\left \langle X, d_1L(n_k,g_{k+1},n_{k+1}) + d_3L (m_k,g_k,n_k) \right \rangle +  \left \langle\xi, \Ad^*_{g_k^{-1}}(\mu _{k})-\mu _{k+1} \right \rangle = 0,
\end{equation}
where, for $1 \leq i \leq 3$, $d_i$ stands for the derivative with respect to the $i$th variable.
\end{example}

\subsection{Discrete dynamics on matched pairs of Lie groupoids}\label{subsect-disc-dynm-on-matched-pair-grpoid}~

\noindent
In this section, we shall rewrite the discrete Euler-Lagrange equations \eqref{eqn-disct-EL} for a matched pair Lie groupoid, that is,
\begin{equation}\label{eqn-matched-discrete-dynm-eqns}
\overleftarrow{U}\left( g_{k},h_k\right) (L)-\overrightarrow{U}\left(g_{k+1},h_{k+1}\right) (L)=0,
\end{equation}
generated by the Lagrangian $L:\C{G}\bowtie\C{H} \to \B{R}$. In view of \eqref{eqn-matched-LIVF} and \eqref{eqn-matched-RIVF}, the equation \ref{eqn-matched-discrete-dynm-eqns} takes the form
\begin{equation} \label{MDD}
\bigg(\overleftarrow{h_k\rt X}(g_k), X^\dagger(h_k)+\overleftarrow{Y}(h_k)\bigg) (L) - \bigg(\overrightarrow{X}(g_{k+1})-Y^\dagger(g_{k+1}),\overrightarrow{Y\lt g_{k+1}}(h_{k+1})\bigg) (L) = 0.
\end{equation}
As such, we arrive at the following proposition.

\begin{proposition} \label{dEL-MP}
Given a matched pair $(\C{G},\C{H})$ of Lie groupoids, the discrete Euler-Lagrange equations on the matched pair groupoid $\C{G}\bowtie\C{H}$ generated by the Lagrangian $L:\C{G}\bowtie\C{H} \to \B{R}$ is given by
 \begin{equation} \label{MDD2}
 \overleftarrow{h_k\rt X}(g_k)(L)-\overrightarrow{X}(g_{k+1})(L) +Y^\dagger(g_{k+1})(L) +  X^\dagger(h_k)(L)+\overleftarrow{Y}(h_k)(L)-\overrightarrow{Y\lt g_{k+1}}(h_{k+1})(L) = 0.
 \end{equation}
\end{proposition} 
 In particular, considering the left action of $\C{H}$ on $\C{G}$ to be trivial, the equation \eqref{MDD2} reduces to 
  \begin{equation}
 \overleftarrow{X}(g_k)(L)-\overrightarrow{X}(g_{k+1})(L) +
 X^\dagger(h_k)(L)+\overleftarrow{Y}(h_k)(L)-\overrightarrow{Y\lt g_{k+1}}(h_{k+1})(L) = 0.
 \end{equation}
 
Similarly, considering this time the right action of $\C{G}$ on $\C{H}$ to be trivial, the equation \eqref{MDD2} takes the form of
 \begin{equation} 
 \overleftarrow{h_k\rt X}(g_k)(L)-\overrightarrow{X}(g_{k+1})(L)+Y^\dagger(g_{k+1})(L) +  \overleftarrow{Y}(h_k)(L)-\overrightarrow{Y}(h_{k+1})(L) = 0.
 \end{equation}
 
If both actions are trivial, then the equation \eqref{MDD2} simplifies to
\[
\overleftarrow{X}(g_k)(L)-\overrightarrow{X}(g_{k+1})(L) + \overleftarrow{Y}(h_k)(L)-\overrightarrow{Y}(h_{k+1})(L) = 0.
\] 

We conclude the present section with yet another particular case of \eqref{eqn-disct-EL} , or of \eqref{eqn-matched-discrete-dynm-eqns}, for a matched pair of Lie groups (regarded as Lie groupoids) to analyse the discerete dynamics on Lie groups from the matched pair point of view.

\subsection{Discrete dynamics on matched pairs of Lie groups}\label{subsect-disc-dynm-matched-Lie-grps}~

\noindent
Let $G$ and $H$ be two Lie groups with mutual actions
\begin{eqnarray} \label{gractions}
\rho:H\times G\rightarrow G,\qquad\left(  h,g\right)  \mapsto
h\rt g, \\\label{hlactions}
\sigma:H\times G\rightarrow H,\qquad\left(  h,g\right)  \mapsto
h\lt g,
\end{eqnarray}
where $h,h_{1},h_{2}\in H$, $g,g_{1},g_{2}\in G$, $e_{H}$ is the identity element in
$H$, and $e_{G}$ is the identity element in
$G$. If the actions \eqref{gractions}-\eqref{hlactions} satisfy
\begin{align}
{\label{condmp}}h\rt\left(  g_{1}g_{2}\right)   &  =\left(
h\rt g_{1}\right)  \left(  \left(  h\lt
g_{1}\right)  \rt g_{2}\right)  ,\\
(h_{1}h_{2})\lt g  &  =\left(  h_{1}\lt\left(
h_{2}\rt g\right)  \right)  \left(  h_{2}\lt
g\right),
\end{align}
then the pair $(G,H)$ is called a mathed pair of Lie groups, \cite{Maji90,Majid-book}. In this case, the cartesian product $G\times H$ may be equipped with the group structure given by the multiplication
\[
(  g_{1},h_{1})  (  g_{2} ,h_{2})   =\left(  g_{1}\left(  h_{1}\rt
g_{2}\right)  ,\left(  h_{1}\lt g_{2}\right)  h_{2}\right)
=\left(  g_{1}\rho\left(  h_{1},g_{2}\right)  ,\sigma\left(  h_{1}%
,g_{2}\right)  h_{2}\right),
\]
and the unit element $\left(e_{G},e_{H}\right)$. This matched pair group is denoted by  $G\bowtie H$. Conversely, if a Lie group $M$ is a cartesian product of two subgroups $G\hookrightarrow M \hookleftarrow H$,
and if the multiplication on $M$ defines a bijection $G\times H\to M$, then
$M$ is a matched pair, that is, $M\cong G\bowtie H$. In this case, the mutual
actions are derived from
\begin{equation}
h\cdot g=\left(  h\rt g\right)  \left(  h\lt
g\right), \label{aux-mutual-actions-group}%
\end{equation}
for any $g\in G$, and any $h\in H$. Let us also record here the inversion in $G\bowtie H$ as
\begin{equation}
( g,h)^{-1}=\left(  h^{-1}\rt g^{-1}%
,h^{-1}\lt g^{-1}\right)  . \label{invelement}%
\end{equation}
for later use.

We next consider the lifting of the group actions to the Lie algebra level. As for the left action, we have
\begin{eqnarray}
H\times \mathfrak{g} \lra \mathfrak{g}, \qquad (h,\xi)\mapsto h\rt \xi:=\dt   h\rt x_t, \label{Hong-g}\\
\mathfrak{h}\times G \lra TG, \qquad (\eta,g) \mapsto \eta^\dagger(g):=\eta \rt g:=\dt   y_t \rt g, \label{honG-g}
\end{eqnarray}
where $\G{g}$ denotes the Lie algebra of the Lie group $G$, as $\G{h}$ stands for the Lie algebra of $H$, $x_t \in G$ is a curve passing through the identity at $t=0$ in the direction of $\xi\in\G{g}$, and finally $y_t\in H$ is a curve passing through the identity in the direction of $\eta\in\G{h}$.

Freezing the group element in \eqref{Hong-g}, we arrive at a linear mapping $h~\rt:\mathfrak{g} \to \G{g}$ for any $h \in H$. We shall denote the transpose of this mapping by $\overset{\ast}{\lt}~ h:\mathfrak{g}^* \to \mathfrak{g}^*$, which is given by 
\begin{equation} \label{*h}
\langle h\rt\xi, \mu\rangle =\langle \xi, \mu \overset{\ast}{\lt} h\rangle,
\end{equation}
for any $\mu \in \mathfrak{g}^*$, where the pairing is the one between $\mathfrak{g}^*$ and $\mathfrak{g}$. 

Similarly, freezing the group element in \eqref{honG-g}, we arrive at a linear operator $\G{b}_g:\mathfrak{h}\mapsto T_gG$, given by $\G{b}_g(\eta):=\eta \rt g$. The transpose of this mapping shall be denoted by $\G{b}_g^*:T_g^*G\mapsto \mathfrak{h}^*$, and it is given by 
\begin{equation} \label{b*}
\langle \eta \rt g, \mu_g \rangle =\langle \G{b}_g(\eta), \mu_g \rangle = \langle \eta, \G{b}_g^*(\mu_g) \rangle,
\end{equation}
for any $\mu_g\in T^*_g G$.

Now on the other hand, for the right $G$ action on $H$, we have the maps
\begin{eqnarray}
\mathfrak{h}\times G \lra \mathfrak{h}, \qquad (\eta,g)\mapsto \eta\lt g:=\dt   y_t\lt g, \label{Gonh-g-}\\
H \times \mathfrak{g} \lra TH, \qquad (h,\xi) \mapsto \xi^\dagger(h):=h\lt\xi:=\dt   h \lt x_t \label{gonH-g-}.
\end{eqnarray}
Similar to above, freezing the group element in \eqref{Gonh-g-} we arrive at a linear mapping $\lt~g:\G{h} \to \G{h}$ for any $g \in G$. The transpose of this map will be denoted by $g~\overset{\ast}{\rt}:\G{h}^\ast \to \G{h}^\ast$, and it is defined by
\begin{equation} \label{*g}
\langle \eta \lt g, \nu \rangle = \langle \eta, g\overset{\ast}{\rt}\nu \rangle,
\end{equation}
for any $\nu\in\mathfrak{h}^*$, where the pairing is the one between $\mathfrak{h}^*$ and  $\mathfrak{h}$. 

Finally, freezing the group element in \eqref{gonH-g-} we obtain a mapping $\G{a}_h:\G{g}\mapsto T_hH$, $\G{a}_h(\xi) = h \lt \xi$ for any $\xi \in \G{g}$. The transpose $\G{a}_h^*:T^*_hH\mapsto \G{g}^*$ of this linear mapping will be given by
\begin{equation} \label{a*}
\langle h\lt \xi,\nu_h \rangle=\langle \G{a}_h(\xi), \nu_h \rangle
= \langle \xi, \G{a}_h^*(\nu_h) \rangle,
\end{equation}
for any $\nu_h\in T^*_hH$.

We note also that if $G\bowtie H$ is a matched pair Lie group, then its Lie algebra is the matched pair Lie algebra $\G{g}\bowtie\G{h}$.  That is, the induced actions  
\[
\rt:\G{h}\ot\G{g}\rightarrow\G{g}
\hbox{ \ \ and \ \ }\lt:\G{h}\ot\G{g}\rightarrow\G{h}
\]
of the Lie algebras satisfy
\begin{equation}
\eta\rt\lbrack\xi_{1},\xi_{2}]=[\eta\rt\xi_{1}%
,\xi_{2}]+[\xi_{1},\eta\rt\xi_{2}]+(\eta\lt\xi
_{1})\rt\xi_{2}-(\eta\lt\xi_{2})\rt
\xi_{1} \label{LAc1}%
\end{equation}
and
\begin{equation}
\lbrack\eta_{1},\eta_{2}]\lt\xi=[\eta_{1},\eta_{2}\lt\xi]+[\eta_{1}\lt\xi
,\eta_{2}]+\eta_{1}\lt(\eta_{2}\rt\xi)-\eta
_{2}\lt(\eta_{1}\rt\xi) \label{LAc2},
\end{equation}
for any $\eta,\eta_{1},\eta_{2}\in\mathfrak{h}$, and any
$\xi,\xi_{1},\xi_{2}\in\mathfrak{g}$. Such a pair $(\mathfrak{g},\mathfrak{h})$ is called a matched pair of Lie algebras, and the Lie algebra structure on $\mathfrak{g\bowtie h}:=\mathfrak{g}
\oplus\mathfrak{h}$ is given by
\begin{equation}
\lbrack(\xi_{1},\eta_{1}),\,(\xi_{2},\eta_{2})]=\left(  [\xi_{1},\xi_{2}%
]+\eta_{1}\rt\xi_{2}-\eta_{2}\rt\xi_{1}%
,\,[\eta_{1},\eta_{2}]+\eta_{1}\lt\xi_{2}-\eta_{2}%
\lt\xi_{1}\right)  . \label{mpla}%
\end{equation}
It is immediate that both $\mathfrak{g}$ and $\mathfrak{h}$ are Lie
subalgebras of $\mathfrak{g}\bowtie\mathfrak{h}$ via the obvious inclusions.
Conversely, given a Lie algebra $\mathfrak{m}$ with two subalgebras
$\mathfrak{g} \hookrightarrow \mathfrak{m} \hookleftarrow \mathfrak{h}$, if $\mathfrak{m}%
\cong \mathfrak{g}\oplus\mathfrak{h}$ via $(\xi,\eta)\mapsto\xi+\eta$, then
$\mathfrak{m}\cong \mathfrak{g}\bowtie\mathfrak{h}$ as Lie algebras. In this case,
the mutual actions of the Lie algebras are uniquely determined by
\[
[\eta,\xi]=(\eta\rt\xi,\,\eta\lt\xi).
\]
On the other hand, there are integrability conditions under which a matched pair of Lie
algebras can be integrated into a matched pair of Lie groups. For a discussion of this
direction we refer the reader to \cite[Sect. 4]{Maji90-II}.

We shall also need the adjoint action of the matched pair Lie group $G \bowtie H$ on its Lie algebra $\G{g}\bowtie \G{h}$. For any $(g,h) \in G \bowtie H$, and any $(\xi,\eta)\in \mathfrak{g}\bowtie \mathfrak{h}$, we recall from \cite[(2.30)]{EsSu16} that
\begin{equation}
\Ad_{(g,h)^{-1}}(\xi,\eta)=(h^{-1}\rt\zeta,T_{h^{-1}}r_h
(h^{-1}\lt\zeta)+\Ad_{h^{-1}}(\eta\lt g)) \label{Ad}%
\end{equation}
where $\zeta := \Ad_{g^{-1}}(\xi)+T_gL_{g^{-1}}(\eta\rt
g) \in \mathfrak{g}$. 

Furthermore, the tangent lifts of the left and right regular actions of $G\bowtie H$ are given in \cite[(2.54)\&(2.55)]{EsenSutl17} as
\begin{align*}
T_{\left(  g_{2},h_{2}\right)  }L_{\left(  g_{1},h_{1}\right)  }\left(
U_{g_{2}},V_{h_{2}}\right) =
\left(  T_{h_{1}\rt g_{2}%
}L_{g_{1}}\left(  h_{1}\rt U_{g_{2}}\right)  ,\,T_{h_{1}%
\lt g_{2}}R_{h_{2}}\left(  h_{1}\lt U_{g_{2}%
}\right)  +T_{{h_{2}}}L_{\left(  h_{1}\lt g_{2}\right)  }%
V_{h_{2}}\right), \nonumber\\ 
T_{\left(  g_{1},h_{1}\right)  }R_{\left(  g_{2},h_{2}\right)  }\left(
U_{g_{1}},V_{h_{1}}\right)  = \left(  T_{g_{1}}R_{\left(  h_{1}%
\rt g_{2}\right)  }U_{g_{1}}+T_{h_{1}\rt g_{2}%
}L_{g_{1}}\left(  V_{h_{1}}\rt g_{2}\right)  ,T_{h_{1}%
\lt g_{2}}R_{h_{2}}\left(  V_{h_{1}}\lt
g_{2}\right)  \right).
\end{align*}
We can thus compute the left and right invariant vector fields generated by a Lie algebra element 
$(\xi,\eta) \in \G{g} \bowtie \G{h}$ as 
\begin{align} \label{rlinvgroup}
 \overleftarrow{(\xi,\eta)}{\left( g,h\right)  }=  T_{\left(  e_G,e_H\right)  }L_{\left(  g,h\right)  }\left(
\xi,\eta\right) =(\overleftarrow{h\rt\xi}(g),h\lt\xi+\overleftarrow{\eta}(h)), \\\label{rrinvgroup}
 \overrightarrow{(\xi,\eta)}{\left( g,h\right)  }=T_{\left(  e_G,e_H\right)  }R_{\left(  g,h\right)  }\left(
\xi,\eta\right) =(\overrightarrow{\xi}(g)+\eta\rt g, \overrightarrow{\eta\lt g}(h)).
\end{align}
Recalling the discrete Euler-Lagrange equations \eqref{discrete-EL-on-Lie-group}, discrete dynamics on $G\bowtie H$ generated by a Lagrangian function $L:G\bowtie H \to \B{R}$ is then given by
\begin{align} \label{MDDG}
\overleftarrow{(\xi,\eta)}{\left( g
_k,h_k\right)  }(L)-\overrightarrow{(\xi,\eta)}{\left( g
_{k+1},h_{k+1}\right)  }(L)=0.
\end{align}
Let now the exterior derivative of the Lagrangian $L:G\bowtie H \to \B{R}$ be a two-tuple $(d_1 L, d_2L)$, where  $d_1 L$ denotes the derivative with respect to group variable $g\in G$ whereas $d_2 L$ denotes the derivative with respect to group variable $h\in H$. Then, in view of the left and right invariant vector fields \eqref{rlinvgroup} - \eqref{rrinvgroup}, we arrive at 
  \begin{align*}
 \left \langle \overleftarrow{h_k\rt\xi}(g_k), d_1 L(g_k,h_k) \right \rangle+
 \left \langle h_k\lt\xi,d_2L(g_k,h_k)\right \rangle +
  \left \langle \overleftarrow{\eta}(h_k), d_2L(g_k,h_k) \right \rangle \\ -
 \left \langle \overrightarrow{\xi}(g_{k+1}),d_1 L(g_{k+1}\right \rangle -
 \left \langle \eta\rt g_{k+1},d_1 L(g_{k+1}) \right \rangle
 - \left \langle \overrightarrow{\eta\lt g_{k+1}}(h_{k+1}), d_2L(g_{k+1},h_{k+1}) \right \rangle =0.
  \end{align*}
It is possible to single out $\xi \in \G{g}$ and $\eta \in \G{h}$ from these equations, that is,
    \begin{align*}
    \Big\langle 
    \xi,\left ( T^*L_{g_k} \cdot d_1L (g_k,h_k)\right)\overset{\ast}{\lt}h_k
    +\G{a}^*_{h_{k}}d_2 L (g_k,h_k)-T^*R_{g_{k+1}} \cdot d_1L (g_{k+1},h_{k+1})
    \Big\rangle \\
    +  \Big\langle
     \eta, T^*L_{h_k}\cdot d_2 L (g_k,h_k)
     -\G{b}^*_{g_{k+1}}d_1 L (g_{k+1},h_{k+1})
     - g_{k+1}\overset{\ast}{\rt} T^*R_{h_{k+1}} \cdot d_2L (g_{k+1},h_{k+1})
    \Big\rangle=0.
     \end{align*}

\begin{proposition} \label{dEL-MP-group}
In particular, taking the covectors  
     \begin{align*}
 T^*R_{g_{k}}\cdot d_1 L (g_{k},h_{k})=\mu_k\in\mathfrak{g}^*, \qquad    
     T^*R_{h_{k}}\cdot d_2 L (g_{k},h_{k})=\nu_k \in\mathfrak{h}^*,
     \end{align*}
the discrete Euler-Lagrange equations on the matched pair Lie group $G \bowtie H$ can be written as
\begin{equation} \label{MDDG-1}
     \Ad^*_{g_k^{-1}}(\mu_k)\overset{\ast}{\lt}h_k+a^*_{h_{k}}d_2 L (g_k,h_k)-\mu_{k+1} + \Ad^*_{h_k^{-1}}(\nu_k)
     -b^*_{g_{k+1}}d_1 L (g_{k+1},h_{k+1})-g_{k+1}\overset{\ast}{\rt} \nu_{k+1}=0.
\end{equation}
\end{proposition} 

Furthermore, when the (right) action of $G$ on $H$ is trivial, we have the discrete Euler-Lagrange equation 
       \begin{equation} \label{MDDG-2}
     \Ad^*_{g_k^{-1}}(\mu_k)\overset{\ast}{\lt}h_k-\mu_{k+1} +  \Ad^*_{h_k^{-1}}(\nu_k) -b^*_{g_{k+1}}d_1 L (g_{k+1},h_{k+1})-\nu_{k+1}=0
       \end{equation}
on the semidirect product Lie group $G\rtimes H$.
      
On the other extreme, assuming the (left) action of $H$ on $G$ to be trivial, we arrive at the equation
\begin{equation} \label{MDDG-3}
     \Ad^*_{g_k^{-1}}(\mu_k)+a^*_{h_{k}}d_2 L (g_k,h_k)-\mu_{k+1}+ 
     \Ad^*_{h_k^{-1}}(\nu_k)-g_{k+1}\overset{\ast}{\rt} \nu_{k+1}=0
       \end{equation}
on the semidirect product Lie group $G\ltimes H$.

If both actions are trivial, then  the equations reduce all the way down to 
\begin{equation} \label{MDDG-4}
     Ad^*_{g_k^{-1}}(\mu_k)-\mu_{k+1} + 
     Ad^*_{h_k^{-1}}(\nu_k)- \nu_{k+1}=0.
\end{equation}

\section{Examples} \label{ex}

\subsection{Discrete dynamics on the trivial groupoid} \label{ex1}~

\noindent
In this subsection we shall illustrate the discrete Euler-Lagrange equation on the trivial groupoid of Example \ref{ex-trivial-groupoid}, regarded as the matched pair groupoid of the coarse (banal) groupoid of Example \ref{ex-coarse-groupoid} and the action groupoid of Example \ref{ex-action-groupoid}. To this end, we first recall from Example \ref{ex-left-inv-right-inv-vf-MG-MM} that
given $(m,g;mg,n) \in (M\times G) \bowtie (M\times M)$ and $(\t_n,\xi;\xi^\dagger(n),Y) \in \C{A}_n((M\times G) \bowtie (M\times M))$, we have
\[
\overleftarrow{(\t_n,\xi;\xi^\dagger(n),Y)} (m,g;mg,n)= (\t_n,\overleftarrow{\xi}(g);\xi^\dagger(mg),\xi^\dagger(n)+X),
\]
and similarly, for any $(\t_m,\xi;\xi^\dagger(m),Y) \in \C{A}_m((M\times G) \bowtie (M\times M))$, 
\[
\overrightarrow{(\t_m,\xi;\xi^\dagger(m),Y)}(m,g;mg,n) = (-\xi^\dagger(m)-Z,\overrightarrow{\xi}(g);-Z\lt g,\t_n).
\]
Now, given $(m_k,g_k;m_kg_k,n_k), (m_{k+1},g_{k+1};m_{k+1}g_{k+1},n_{k+1}) \in (M\times G) \bowtie (M\times M)$, so that 
\[
\b((m_k,g_k;m_kg_k,n_k)) = n_k = m_{k+1} = \a(m_{k+1},g_{k+1};m_{k+1}g_{k+1},n_{k+1}),
\]
and a Lagrangian $L:(M\times G) \bowtie (M\times M) \to \B{R}$, the equation \eqref{MDD2} yields
\begin{align}\label{discrete-tirivial-matched-pair-I}
\begin{split}
& \left \langle X(n_{k}),d_1L(n_k,g_{k+1};m_{k+1}g_{k+1},n_{k+1})+d_4L (m_k,g_k;m_kg_k,n_k)\right\rangle  +\\ 
& \left\langle \xi, Ad^*_{g_k^{-1}}\mu _{k}-\mu _{k+1} \right\rangle
+ \left \langle \xi^\dagger(m_kg_k),d_3L(m_k,g_k;m_kg_k,n_k)\right\rangle  \\
& \left \langle \xi^\dagger(n_{k}),d_1L(n_k,g_{k+1};m_{k+1}g_{k+1},n_{k+1})+d_4L (m_k,g_k;m_kg_k,n_k)\right\rangle + \\
& \left \langle X(n_{k})\lt g_{k+1},d_3L(n_k,g_{k+1};m_{k+1}g_{k+1},n_{k+1})\right\rangle =0
\end{split}
\end{align}
where $\mu _{k}=T^*r_{g_k}d_2L(m_k,g_k;m_kg_k,n_k)$, and for $1\leq j \leq 4$, the operator $d_j$ denotes the derivative with respect to the $j$th variable.

\begin{remark}
Let us note that the equations \eqref{discrete-tirivial-matched-pair-I} above correspond to the discrete Euler-Lagrange equations \eqref{discrete-tirivial-I} on the trivial groupoid $M\times G \times M$, under the isomorphism \eqref{map-Phi}. More precisely, on one hand we have
\begin{align*}
& \overleftarrow{(\t_n,\xi;\xi^\dagger(n),Y)} (m,g;mg,n)= (\t_m,\overleftarrow{\xi}(g);\xi^\dagger(mg),\xi^\dagger(n)+X) = \\
& \dt\, (m,ge_t;mge_t,n_te_t) = \dt\,\Phi(m,ge_t,n_te_t) = T\Phi\left(\t_m,\overleftarrow{\xi}(g),\xi^\dagger(n)+X\right) = \\
& T\Phi\left(\overleftarrow{(\t_n,\xi,\xi^\dagger(n)+X)}(m,g,n)\right),
\end{align*}
where on the other hand,
\begin{align*}
& \overrightarrow{(\t_m,\xi;\xi^\dagger(m),Y)}(m,g;mg,n) = (-\xi^\dagger(m)-Z,\overrightarrow{\xi}(g);-Z\lt g,\t_n) = \\
& - \dt\, (m_te_t,e_t^{-1}g;m_tg,n) = -\dt\, \Phi(m_te_t,e_t^{-1}g,n) = T\Phi\left(-\xi^\dagger(m)-Z,\overrightarrow{\xi}(g),\t_n\right) =\\
& T\Phi\left(\overrightarrow{(\t_m,\xi,\xi^\dagger(m)+Z)}(m,g,n)\right).
\end{align*}
The correspondence, then, follows at once.
\end{remark}

\subsection{Discrete Dynamics on $SL(2,\mathbb{C})=SU(2)\bowtie K$}\label{ex2}~

\noindent
In this subsection, we shall study the discrete Euler-Lagrange equations on the Lie group $SL(2,\mathbb{C})$ from the matched pair point of view. To this end, we shall first recall its decomposition
\begin{equation}\label{ExMPLG}
SL(2,\mathbb{C}) =SU(2)\bowtie K
\end{equation}
from \cite{Maji90-II}, see also \cite{EsSu16,EsenSutl17}, the group structures, the mutual actions of the groups $SU(2)$ and $K$, together with their lifts. 

The group
\begin{equation}
SU(2)=\left\{
\begin{pmatrix}
\omega & \vartheta\\
-\bar{\vartheta} & \bar{\omega}%
\end{pmatrix}
\in SL(2,\mathbb{C}):\left\vert \omega\right\vert ^{2}+\left\vert
\vartheta\right\vert ^{2}=1\right\}  \label{SU(2)-1}
\end{equation}
in the matched pair decomposition\eqref{ExMPLG} is a universal
double cover of the group $SO\left( 3\right) $. As such, for each element $A\in SU(2)$ there exists a unique matrix ${\rm Rot}_A\in SO(3)$. The Lie algebra $\mathfrak{su}(2)$ of the group $SU(2)$ is the matrix Lie algebra
\[
\mathfrak{su}(2)=\left\{  \frac{-\iota}{2}\left(
\begin{array}
[c]{cr}%
t & r-\iota s\\
r+\iota s & -t
\end{array}
\right)  :r,s,t\in\mathbb{R}\right\}
\]
of traceless skew-hermitian matrices. Following \cite{Maji90-II} we fix three matrices
\begin{equation}
e_{1}=\left(
\begin{array}
[c]{cc}%
0 & -\iota/2\\
-\iota/2 & 0
\end{array}
\right)  ,\,\,e_{2}=\left(
\begin{array}
[c]{cc}%
0 & -1/2\\
1/2 & 0
\end{array}
\right)  ,\,\,e_{3}=\left(
\begin{array}
[c]{cc}%
-\iota/2 & 0\\
0 & \iota/2
\end{array}
\right)  \label{basissu(2)}%
\end{equation}
as a basis of the Lie algebra $\mathfrak{su}(2)$. We further make use of this to identify
the matrix Lie algebra $\mathfrak{su}(2)$ with the Lie algebra $\mathbb{R}^{3}$ by the cross product;
\begin{equation}
re_{1}+se_{2}+te_{3}\in\mathfrak{su}(2)\longleftrightarrow\mathbf{X}=\left(
r,s,t\right)  \in\mathbb{R}^{3}. \label{su(2)toR}
\end{equation}
We also identify the dual space $\mathfrak{su}(2)^{\ast}$ of $\mathfrak{su}%
(2)\cong\mathbb{R}^{3}$ with $\mathbb{R}^{3}$ using the Euclidean dot product. Using
this dualization, we can express the coadjoint action of the Lie algebra
$\mathfrak{su}(2)\cong\mathbb{R}^{3}$ on
$\mathfrak{su}^{\ast}(2)\cong\mathbb{R}^{3}$ as%
\begin{equation}
\ad^{\ast}: \mathfrak{su}(2)\times\mathfrak{su}^{\ast}(2)\rightarrow
\mathfrak{su}^{\ast}(2),\text{ \ \ } (\mathbf{X},\mathbf{\Phi}) \mapsto
\ad_{\mathbf{X}}^{\ast}\mathbf{\Phi} := \mathbf{X}\times\mathbf{\Phi}, {\label{ad*1}}%
\end{equation}
for any $\mathbf{X}\in\mathfrak{su}(2)\cong\mathbb{R}^{3}$, and any $\mathbf{\Phi}%
\in\mathfrak{su}^{\ast}(2)\simeq\mathbb{R}^{3}$.

The simply-connected group  $K$, on the other hand, may be represented by
\begin{equation}
K=\left\{  \frac{1}{\sqrt{1+c}}%
\begin{pmatrix}
1+c & 0\\
a+ib & 1
\end{pmatrix}
\in SL(2,\mathbb{C})\mid a,b\in\mathbb{R} \text{ and }c>-1\right\}  \label{K-1}
\end{equation}
where the group operation is the matrix multiplication. The Lie algebra $\mathfrak{K}$ of the group $K$ is thus given by
\begin{equation}
\mathfrak{K}=\left\{  \left(
\begin{array}
[c]{cc}%
\frac{1}{2}c & 0\\
a+i b & \frac{-1}{2}c
\end{array}
\right)  \in\mathfrak{sl}\left(  2,\mathbb{C}\right)  \mid a,b,c\in
\mathbb{R}\right\}  \label{k-1}%
\end{equation}
with matrix commutator being the Lie bracket.
The group $K$ can also be realised as a subgroup of $GL(3,\mathbb{R})$ as
\begin{equation}\label{K-2}
K=\left\{
\begin{pmatrix}
1+c & 0 & 0\\
0 & 1+c & 0\\
-a & -b & 1
\end{pmatrix}
\in GL\left(  3,\mathbb{R}\right) \mid a,b\in \mathbb{R} \text{ and }c>-1\right\},
\end{equation}
where the group operation is the matrix multiplication. In this case, its Lie algebra $\mathfrak{K}$ is given by
\begin{equation}
\mathfrak{K}=\left\{  \left(
\begin{array}
[c]{ccc}%
c & 0 & 0\\
0 & c & 0\\
-a & -b & 0
\end{array}
\right)  \in\mathfrak{gl}(3,\mathbb{R})\mid a,b,c\in\mathbb{R}\right\},
\label{k-2}%
\end{equation}
where the Lie bracket is the matrix commutator. The group $K$ can, alternatively, be identified with the
subspace
\begin{equation}
K=\left\{  (a,b,c)\in%
\mathbb{R}
^{3} \mid a,b\in%
\mathbb{R}
\text{ and }c>-1\right\}   \label{K-3}%
\end{equation}
of $%
\mathbb{R}^{3}$ with a non-standard multiplication
\begin{equation*}
(a_{1},b_{1},c_{1})\ast(a_{2},b_{2},c_{2})=(a_{1},b_{1},c_{1})(1+c_{2}%
)+(a_{2},b_{2},c_{2}),
\end{equation*}
in which case the Lie algebra $\mathfrak{K}$ is $\mathbb{R}^{3}$ via the Lie bracket 
\begin{equation}
\lbrack\mathbf{Y}_{1},\mathbf{Y}_{2}] = \mathbf{k}\times
(\mathbf{Y}_{1}\times\mathbf{Y}_{2}), \label{k-3}%
\end{equation}
where $\mathbf{k}$ is the unit vector $(0,0,1) \in \B{R}^3$. In this case, using the dot product, we may identify the dual space $\mathfrak{K}^{\ast}$  with $\mathbb{R}^3$ as well. Then, the coadjoint action of the Lie algebra $\mathfrak{K}%
\cong\mathbb{R}^{3}$ on its dual space $\mathfrak{K}^{\ast} \cong \mathbb{R}^3$ can be computed as
\begin{equation}
\ad^{\ast}: \mathfrak{K}\times\mathfrak{K}^{\ast}\rightarrow\mathfrak{K}^{\ast
},\text{ \ \ } (\mathbf{Y},\mathbf{\Psi})\mapsto \ad_{\mathbf{Y}}^{\ast
}\mathbf{\Psi}:=\left(  \mathbf{k}\cdot\mathbf{Y}\right)  \mathbf{\Psi}-\left(
\mathbf{\Psi}\cdot\mathbf{Y}\right)  \mathbf{k}, {\label{ad*2}}%
\end{equation}
for any $\mathbf{Y}\in\mathfrak{K}\simeq\mathbb{R}^{3}$, and any $\mathbf{\Psi}%
\in\mathfrak{K}^{\ast}\simeq\mathbb{R}^{3}$. The group isomorphisms relating (\ref{K-1}), (\ref{K-2}), and (\ref{K-3}) are
given by
\begin{equation}
\frac{1}{\sqrt{1+c}}%
\begin{pmatrix}
1+c & 0\\
a+ib & 1
\end{pmatrix}
\longleftrightarrow%
\begin{pmatrix}
1+c & 0 & 0\\
0 & 1+c & 0\\
-a & -b & 1
\end{pmatrix}
\longleftrightarrow(a,b,c). \label{Griso}%
\end{equation}
The Lie algebra isomorphisms
\begin{equation}
\left(
\begin{array}
[c]{cc}%
\frac{1}{2}c & 0\\
a+\iota b & \frac{-1}{2}c
\end{array}
\right)  \longleftrightarrow\left(
\begin{array}
[c]{ccc}%
c & 0 & 0\\
0 & c & 0\\
-a & -b & 0
\end{array}
\right)  \longleftrightarrow\left(  a,b,c\right)  \label{ktoR}%
\end{equation}
between (\ref{k-1}), (\ref{k-2}) and (\ref{k-3}) are then obtained by
differentiating (\ref{Griso}).

We now move on to the mutual actions of the groups $SU(2)$ and $K$ on each other. 
Given any $A\in SU(2)$, and any $B\in K\subset SL(2,\mathbb{C})$, the left action of $K$ on $SU(2)$ is given by
\begin{equation}\label{KonSU(2)}
B\rt A   =\left\Vert BA\left(
\begin{array}
[c]{cc}%
0 & 0\\
0 & 1
\end{array}
\right)  \right\Vert _{M}^{-1}\left(  BA\left(
\begin{array}
[c]{cc}%
0 & 0\\
0 & 1
\end{array}
\right)  +{B^{-\dagger}}A\left(
\begin{array}
[c]{cc}%
1 & 0\\
0 & 0
\end{array}
\right)  \right),
\end{equation}
where $B^{-\dagger}$ stands for the inverse of the conjugate
transpose of $B \in K$, and $\left\Vert B\right\Vert _{M}^{2}=tr(B^{\dagger}B)$ refers to
the matrix norm on $SL(2,\mathbb{C})$. The right action of $SU(2)$ on $K \subset \mathbb{R}^3$, on the other hand, is
\begin{equation}\label{SU(2)onK}
B\lt A = \frac{\left\Vert {\mathbf B} \right\Vert _{E}^{2}}{2\left(  c+1\right)  }e_{3}+A\left(  B-\frac{\left\Vert {\mathbf B} \right\Vert
_{E}^{2}}{2\left(  c+1\right)  }e_{3}\right)  A^{-1},
\end{equation}
where $\left\Vert {\bullet} \right\Vert _{E}^{2}:\mathbb{R}^{3} \to  \mathbb{R}$ denotes the Euclidean norm, in view of the identification (\ref{Griso}) of $B \in K \subset SL(2,\mathbb{C})$ with ${\mathbf B} \in K \subset \mathbb{R}^3$.

Differentiating (\ref{KonSU(2)}) with respect to $A \in SU(2)$, and
regarding $B\in K\subset GL(3,\mathbb{R})$ via (\ref{Griso}), we obtain
\begin{equation}\label{Konsu(2)}
\mathbf{\rt}:K\times\mathfrak{su}\left(  2\right)
\rightarrow\mathfrak{su}\left(  2\right)  \text{, \ \ }\left(  B,\mathbf{X}%
\right)  \mapsto B\rt\mathbf{X} := B\mathbf{X},
\end{equation}
for any $\mathbf{X}\in \mathfrak{su} \left(  2\right)  \cong\mathbb{R}^{3}$. 
Freezing the group element here, we get a linear operator $B\rt :\mathfrak{su}\left(  2\right) \to \mathfrak{su}\left(  2\right)$. The transpose of this operator $\overset{\ast}{\rt}B : \mathfrak{su}^\ast\left(  2\right) \to \mathfrak{su}^\ast\left(  2\right)$ is given by
\begin{equation} \label{ex-B*}
\Phi\overset{\ast}{\lt}B=B^T\Phi.
\end{equation}
Similarly, the derivative of \eqref{SU(2)onK} with respect to $A\in SU(2)$ renders the infinitesimal right action of the Lie algebra $\mathfrak{su}\left(  2\right)$ on $K$ as
\begin{equation}
\lt:K\times\mathfrak{su}\left(  2\right)  \rightarrow TK,\qquad
(B, {\mathbf X})\mapsto B\lt\mathbf{X} = T_{e_K}r_B\left(  \mathbf{X} \times\mathbf{\widetilde{B}}\right)  , \label{XonB}%
\end{equation}
where $\mathbf{X}\in\mathfrak{su}\left(  2\right)  \cong\mathbb{R}^{3}$, and
\[
\mathbf{\widetilde{B}} := \frac{1}{c+1}\mathbf{B}-\frac{\left\Vert
\mathbf{B}\right\Vert _{E}^{2}}{2(c+1)^{2}}\mathbf{k} 
\]
identifying once again  $B \in K \subset SL(2,\mathbb{C})$ with
${\mathbf B} \in K \subset \mathbb{R}^3$ via (\ref{Griso}). Here, $T_{e_K}r_B$
is the tangent lift of the right translation $r_B:K\rightarrow K$ by $B\in K$,
and it acts simply by the matrix multiplication regarding
$\mathbf{X} \times\mathbf{\widetilde{B}} \in \mathfrak{K} \cong \mathbb{R}^3 \cong \mathfrak{gl}(3,\mathbb{R})$
via (\ref{ktoR}). Freezing the group element in \eqref{XonB} we arrive at a linear operator $\G{a}_B:\mathfrak{su}\left(  2\right)\to T_BK$, the transpose of which is the operator $\G{a}^\ast_B:T^*_BK\to \mathfrak{su}^\ast\left(  2\right)$ given by 
\begin{equation} \label{ex-a*}
\G{a}_B^*(\Psi_B)=T^*r_B(\Psi_B)\times \mathbf{\widetilde{B}}
\end{equation}
for any $\Psi_B\in T^*_BK$. 

Next, the derivative of \eqref{KonSU(2)} with respect to $B \in K$ at the identity, in the direction of $\mathbf{Y}\in \G{K}\subset \mathbb{R}^3$, yields
\begin{align} \label{KonSU}
\rt:\G{K}\times SU(2) \to TSU(2), \qquad {\bf Y} \rt A = Tr_A\Big({\bf Y} \times \big(\Ad_A(e_3) - e_3\big)\Big),
\end{align}
where we consider $\Ad_A(e_3) - e_3 \in \B{R}^3$ to perform the vector product, then we view the resulting element in $SU(2)$, \ie as a $2\times2$ complex matrix. Freezing the group element in \eqref{KonSU}, we obtain a mapping $\G{b}_A:\G{K}\to T_ASU(2)$. The transpose of this operator $\G{b}_A^*:T^*_ASU(2)\to \G{K}^*$ is given explicitly by
\begin{equation} \label{ex-b*}
\G{b}_A^*(\Phi_A)=\big(\Ad_A(e_3) - e_3\big)\times Tr_A^* \Phi_A.
\end{equation}
Similarly, the derivative of \eqref{SU(2)onK} with respect to $B \in K$ in the direction of ${\bf Y} \in \G{K} \cong \B{R}^3$ produces 
\begin{align*}
\lt:\G{K}\times SU(2) \to \G{K}, \qquad {\bf Y} \lt A = {\rm Rot}_A({\bf Y}),
\end{align*}
and hence defines a linear mapping $\lt A:\G{K} \to \G{K}$, whose transpose $A \overset{\ast}{\rt} :\G{K}^\ast \to \G{K}^\ast$ may be given by
\begin{equation} \label{ex-A*}
A \overset{\ast}{\rt} \Psi := {\rm Rot}_A^* \Psi,
\end{equation}
for any $\Psi \in \G{K}^\ast$.

Now we are ready to write the discrete Euler-Lagrange equations on the matched pair Lie group $SL(2,\mathbb{C})=SU(2)\bowtie K$. Substituting\eqref{ex-B*}, \eqref{ex-a*}, \eqref{ex-b*}, and \eqref{ex-A*} into \eqref{MDDG-1}, we conclude that  
\begin{align*}
\begin{split}
& B_k^T(\Ad^*_{A_k}\Phi_k)+T^*r{_{B_k}}(d_2L(A_k,B_k))\times \mathbf{\widetilde{B}}-\Phi_{k+1} + \\
& \Ad^*_{B_k}\Psi_k- \big(\Ad_{A_{k+1}}(e_3) - e_3\big) \times T^*r_{A_{k+1}}(d_1L(A_{k+1},B_{k+1}))
-{\rm Rot}_{A_{k+1}}^*\Psi_{k+1}=0.
\end{split}
\end{align*} 

\section{Conclusion and Discussions}

In the present paper, we have studied the discrete dynamics on the matched pairs of Lie groupoids. This enabled us to study the equations of motion governing two mutually interacting discrete systems. More precisely, we have presented the discrete Euler-Lagrange equations on a matched pair Lie groupoid as a sum of those over the individual Lie groupoids which are matched - but enriched with the mutual actions. Taking advantage of the Lie groups being the quintessential examples of Lie groupoids, in particular we have introduced the discrete Euler-Lagrange equations on (matched pair) Lie groups. 

In order to illustrate the theory, we have studied two concrete examples; the trivial Lie groupoid, and the group $SL(2,\mathbb{C})$. 

The trivial Lie groupoid is a matched pair of the action groupoid and the coarse groupoid. In view of this observation we have realised the discrete Euler-Lagrange equations of the trivial Lie groupoid as a sum of the discrete Euler-Lagrange equations of the action groupoid, and the Euler-Lagrange equations of the coarse groupoid, but decorated with the mutual action terms.

As for the second tangible example, we considered the matched pair decomposition of the group $SL(2,\mathbb{C})$ into $SU(2)$ and a simply connected subgroup $K\subseteq SL(2,\mathbb{C})$, which, in turn, corresponds to its Iwasawa decomposition. We then similarly decomposed the discrete Euler-Lagrange equations over $SL(2,\mathbb{C})$ into those over $SU(2)$ and $K$, once again, furnished with the action terms.

We finally note that, the present paper concerns only the discrete dynamics generated by Lagrangian functions on Lie groupoids. It is very well known that there exists a theory of Lagrangian dynamics on the Lie algebroid level as well; \cite{martinez2001lagrangian,Wein96}. We plan to apply the matched pair strategy to the Lagrangian dynamics from the point of view of the Lie algebroids, which however, deserves a separate paper.

\section{Acknowledgments}
The first named author (OE) is greateful to Prof. Manuel de Le\'on
 and Prof. David Mart\'in de Diego for their encouragements and kind interests in this present study.  
Both authors acknowledge the support by T\"UB\.ITAK (the Scientific and Technological Research Council of Turkey) under the project "Matched pairs of Lagrangian and Hamiltonian Systems" with the project number 117F426.
\bibliographystyle{plain}
\bibliography{references}{}

\end{document}